\documentclass[sigconf]{aamas} 
\usepackage{balance}
\usepackage{graphicx}
\usepackage{latexsym}
\usepackage{amsfonts}
\usepackage{color}
\usepackage{amsmath}
\usepackage[linesnumbered,ruled]{algorithm2e}
\usepackage{enumitem}
\usepackage{tabularx}
\usepackage{stfloats}
\usepackage{multirow}
\usepackage{threeparttable}
\usepackage{makecell}
\usepackage{bm}
\usepackage{nicefrac}
\usepackage{booktabs}
\usepackage{array}

\newtheorem{definition}{Definition}
\newtheorem{observation}{Observation}
\newtheorem{lemma}{Lemma}

\DeclareMathOperator*{\argmax}{arg\,max}
\setcopyright{ifaamas}
\copyrightyear{2024}
\acmYear{2024}
\acmDOI{}
\acmPrice{}
\acmISBN{}

\title{Viral Marketing in Social Networks with Competing Products}

\author{Ahad N. Zehmakan}
\affiliation{
  \institution{The Australian National University}
  \city{}
  \country{}}
\email{ahadn.zehmakan@anu.edu.au}

\author{Xiaotian Zhou}
\affiliation{
  \institution{Fudan University}
  \city{}
  \country{}}
\email{22110240080@m.fudan.edu.cn}

\author{Zhongzhi Zhang}
\affiliation{
  \institution{Fudan University}
  \city{}
  \country{}}
\email{zhangzz@fudan.edu.cn}

\begin{abstract}
Consider a directed network where each node is either red (using the red product), blue (using the blue product), or uncolored (undecided). Then in each round, an uncolored node chooses red (resp. blue) with some probability proportional to the number of its red (resp. blue) out-neighbors.

What is the best strategy to maximize the expected final number of red nodes given the budget to select $k$ red seed nodes? After proving that this problem is computationally hard, we provide a polynomial time approximation algorithm with the best possible approximation guarantee, building on the monotonicity and submodularity of the objective function and exploiting the Monte Carlo method. Furthermore, our experiments on various real-world and synthetic networks demonstrate that our proposed algorithm outperforms other algorithms.

Additionally, we investigate the convergence time of the aforementioned process both theoretically and experimentally. In particular, we prove several tight bounds on the convergence time in terms of different graph parameters, such as the number of nodes/edges, maximum out-degree and diameter, by developing novel proof techniques.
\end{abstract}

\newcommand{\BibTeX}{\rm B\kern-.05em{\sc i\kern-.025em b}\kern-.08em\TeX}

\begin{document}

\pagestyle{fancy}
\fancyhead{}

%%% The next command prints the information defined in the preamble.

\maketitle 

\section{Introduction}
The emergence of online social networks, such as Facebook, Twitter, WeChat, and Instagram, has precipitated a paradigm shift in communication methods in the 21st century. These digital platforms have fundamentally transformed the manner in which individuals interact and exchange information. Of particular significance is the profound influence online social networks exert on the formation of opinions and dissemination of information.

The proliferation of online social networking sites, coupled with advancements in information technology, has sparked a keen interest in leveraging social networks to advertise new products or promote political campaigns. Nowadays, many firms are opting out of utilizing online social networks for advertising purposes, in lieu of more traditional methods.

% Companies regularly attempt to convince a group of consumers within the social network to adopt their new products, either through targeted advertising, providing free samples, or adding monetary incentive. If these people influence some of their friends to also try the product, and these friends in turn recommend it to others, and so forth, the producer can create a cascade of recommendations. This is usually known as viral marketing~\cite{kempe2003maximizing} which is a commonly used techniques to promote new products.

Corporations frequently employ diverse strategies to persuade a specific segment of consumers on social media platforms to adopt their new products, such as targeted advertising, providing free samples, or monetary incentives. By harnessing the influence of these individuals and encouraging them to recommend the product to their social circles, a chain reaction of recommendations can be created, usually referred to as viral marketing, cf.~\cite{lin2015analyzing, myers2012clash}. This technique has emerged as a prominent method for promoting new products, as it enables companies to achieve extensive reach and exposure while keeping costs low.

The question then becomes how to choose an initial subset of so-called early adopters to maximize the number of people that will eventually be reached, given some fixed marketing budget.
% . The size of the subsets allowed is assumed to be limited due to marketing budget constraints.
To tackle this question, several stochastic models, such as the Independent Cascade (IC) and Linear Threshold (LT) model~\cite{kempe2003maximizing}, have been introduced to simulate the product adoption process. In most of these models, one considers a graph where each node is either colored (active) or uncolored (inactive). Then, in each round of the process, some uncolored nodes become colored following a predefined stochastic updating rule. The nodes correspond to individuals, and the edges represent relationships such as friendship, interest, or collaboration. A node is colored when it has adopted the product, and the updating rule defines the way the adoption progresses.

The problem of finding a seed set of size $k$ which maximizes the expected final number of colored nodes has turned out to be NP-hard in various setups, cf.~\cite{kempe2003maximizing, li2018influence, bredereck2017manipulating, tao2022hard}; however, several greedy-based and centrality-based approximation algorithms have been developed, cf.~\cite{li2018influence,chen2010scalable}. Another aspect of diffusion processes which has been studied extensively is the convergence time, where tight bounds in terms of different network parameters or for special classes of networks have been provided, cf.~\cite{frischknecht2013convergence, lesfari2022biased, berenbrink2022asynchronous,ref3}.

The extant body of prior work has predominantly centered on single cascade models, cf.~\cite{li2018influence, chen2010scalable, wilder2018controlling}, albeit this presupposition proves inadequate in numerous practical situations, particularly in the presence of multiple rival products. It is not uncommon for creators of consumer technologies to introduce a new product in a market where a competitor is presenting a comparable product. Consequently, several extensions of the single cascade models such as the IC and LT model have been introduced to capture the multiple cascade framework, cf.~\cite{lin2015analyzing, myers2012clash, bharathi2007competitive, wang2021maximizing, liu2016containment, pathak2010generalized, wu2015maximizing, auletta2020effectiveness}. Given a state where nodes are either uncolored (undecided) or blue (use the blue product), the objective is to maximize the final reach of red color (adoption of the red product) with a given budget to make $k$ uncolored nodes red. It is often presumed that nodes do not switch between red and blue color as this may entail incurring a transition cost that could outweigh the direct advantages of the competitive technology, cf.~\cite{farrell1986installed}. Furthermore, the red company is aware of its competitor’s early adopters, for example through extensive market research or industrial espionage.

In contrast to single cascade, our knowledge of the multiple cascade framework is constrained. Therefore, gaining a more profound understanding of the driving mechanisms behind multiple products adoption processes is very fundamental. This is especially critical given that the limited work done on this topic have employed a contrived generalization of existing single cascade models rather than devising models tailored for the multiple products' setup. To close this gap, the present work develops a natural, simple, and intuitive model for the adoption of multiple products. We draw inspiration from the rich literature on opinion formation models, such as the Majority model~\cite{zehmakan2020opinion} and Voter model~\cite{hassin2001distributed}, in which each node has an opinion from a range of opinions and updates its opinion through interactions with its peers. In our model, called \textsc{Random Pick}, nodes are uncolored, red, or blue and in each round an uncolored node picks one of its out-neighbors at random and adopts its color. This process particularly possesses the property that the probability of an uncolored node $v$ adopting red (resp. blue) color is proportional to the number of red (resp. blue) nodes in its out-neighborhood.

We prove that the problem of maximizing the expected final number of red nodes by selecting $k$ red seed nodes in the \textsc{Random Pick} model is computationally hard and provide a polynomial time approximation algorithm with theoretical guarantees. We also give several tight bounds on the convergence time of the \textsc{Random Pick} process in terms of different graph parameters, in both worst case and average setups. We complement our theoretical findings with a large set of experiments on real-world and synthetic graph data.

% \textbf{Roadmap.} For the remainder of this section, we furnish a few elementary definitions, summarize our findings briefly, and present a survey of related previous research. In Section~\ref{algo-sec}, our inapproximability results and algorithmic findings are presented and in Section~\ref{sec:conv} our bounds on the convergence time are proven.

% \vspace{-0.2cm}
\subsection{Basic Definitions}
\label{preliminaries-sec}
\textbf{Graph Definitions.} Consider a directed graph $G=(V,E)$, with the node set $V$ and edge set $E\subseteq V\times V$. Let $n:=|V|$ and $m:=|E|$. For a node $v\in V$, let $\Gamma_{+}(v):=\{u:(v,u)\in E\}$ and $\Gamma_{-}(v):=\{u:(u,v)\in E\}$ be the set of \emph{out-neighbors} and \emph{in-neighbors} of $v$. Furthermore, $d_+(v):=|\Gamma_+(v)|$ and $d_-(v):=|\Gamma_-(v)|$ denote the \textit{out-degree} and \textit{in-degree} of node $v$. We define $\Delta_{+}(G):=\max_{v\in V} d_+(v)$ to be the maximum out-degree of $G$. The node sequence $v_1,\cdots, v_k$ is a \textit{walk} if $(v_i,v_{i+1})\in E$ for each $1\le i\le k-1$. A \textit{path} is a walk in which all nodes are distinct. The \textit{distance} $d(v,u)$ from $v$ to $u$ is the length of the shortest path from $v$ to $u$. The distance $d(v,u)=\infty$ if there is no path from $v$ to $u$, i.e., $u$ is not \textit{reachable} from $v$. The \textit{diameter} of $G$ is denoted by $D(G):=\max_{v,u\in V, d(v,u)\ne \infty}d(v,u)$. (Note that we exclude the unreachable pairs.) We simply use $D$ and $\Delta_+$, when $G$ is clear from the context.

Assume that in $G$, if $(v,u)\in E$, then $(u,v)\in E$. Then, we say that $G$ is \textit{undirected}. In that case, we simply use $\Gamma(v)$ instead of $\Gamma_-(v)$ and $\Gamma_+(v)$ and use $\Delta$ instead of $\Delta_+$.

\noindent \textbf{Model Definition.} A \textit{state} is a function $\mathcal{S}:V\rightarrow \{b,r,u\}$, where $b$, $r$, and $u$ represent \textit{blue}, \textit{red}, and \textit{uncolored}, respectively. We say a node is \emph{colored} if it is blue/red.

% \vspace{-0.1cm}
\begin{definition}[Random Pick]
\label{random-pick-def}
In the \textsc{Random Pick} model for a given initial state $\mathcal{S}_0$, in each round every node $v$ picks an out-neighbor $w$ (i.e., a node in $\Gamma_+(v)$) uniformly and independently at random; then, $v$ adopts $w$'s color if $v$ is uncolored and $w$ is colored. See Figure~\ref{model-example} for an example.
\end{definition}

% \vspace{-0.5cm}
\begin{figure}[t]
\centerline{\includegraphics[height=1.0in]{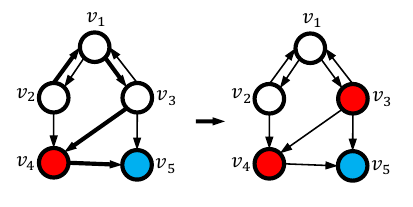}}
\caption{One application of \textsc{Random Pick} model, where bold edges show the picked random out-neighbor for each node.} \label{model-example}
\vspace{-0.5cm}
\end{figure}

Note that the red/blue nodes and the uncolored nodes with no blue/red out-neighbors remain unchanged, regardless of their random choice of out-neighbor. Thus, we could redefine the model such that only uncolored nodes with at least one colored node choose a random out-neighbor. However, the model is defined as described in Definition~\ref{random-pick-def} since it makes our analysis more straightforward.

% We observe that t
The probability of an uncolored node $v$ adopting red/blue color is proportional to the number of red/blue nodes in $\Gamma_+(v)$. Furthermore, we could define the model in a \textit{push-based} manner (rather than \textit{pull-based}), where each node pushes its color to all of its in-neighbors and then each uncolored node picks one of the received colors uniformly at random. This description perhaps matches the reality more accurately, but the described models are identical.

Let $\mathcal{S}_t$, for $t\ge 0$, denote the state in the $t$-th round. Furthermore, we define $\mathcal{S}^{b}$, $\mathcal{S}^{r}$, and $\mathcal{S}^{u}$ to be the set of red/blue/uncolored nodes in the state $\mathcal{S}$. We define $R_t=\mathcal{S}_t^r$, $B_t=\mathcal{S}_t^b$ and $U_t=\mathcal{S}_t^u$ to be the set of red, blue, and uncolored nodes in $\mathcal{S}_t$. Set $r_t:=|R_t|$, $b_t:=|B_t|$, and $u_t:=|U_t|$. Note that they are all random variables.

\textbf{Convergence Properties.} %We observe that a
A node $v\in U_0$ will be \textit{eventually} colored red/blue if and only if it can reach a node $w\in R_0\cup B_0$. Thus, the \textsc{Random Pick} process will eventually reach a \textit{stable} state, where no node can update, that is, no uncolored node can be colored. (In other words, in the corresponding Markov chain there is a path from every state to some stable state and the stable states are absorbing.) The number of rounds the process needs to reach such a stable state is called the \textit{convergence time} of the process.

\textbf{Pick Sequence.} Let $ps(v)\in V\times V\times \cdots$ denote the \emph{pick sequence} for a node $v$, where the $t$-th element of the sequence is the random out-neighbor picked by node $v$ in the $t$-th round. We use $ps_t(v)$ to denote the $t$-th element of the sequence. Let $\mathcal{P}(G):=(ps(v_1),\cdots, ps(v_n))$ be the \emph{pick profile} of $G=(V=\{v_1,\cdots,v_n\},E)$. We observe that given an initial state $\mathcal{S}_0$ and a pick profile $\mathcal{P}$, the final color of all nodes can be inferred deterministically.

\textbf{Product Adoption Maximization.} For a state $\mathcal{S}$ which contains only blue/uncolored nodes (i.e., the red product has not entered the market yet) and a set $A\subset \mathcal{S}^u=V\setminus \mathcal{S}^b$, let $\mathcal{F}_{\mathcal{S}}(A)$ be the expected final number of red nodes in the \textsc{Random Pick} process starting from the state $\mathcal{S}'$ obtained from setting the colors of all nodes in $A$ to red in $\mathcal{S}$.

% \vspace{-0.1cm}
\begin{definition}[\textsc{Product Adoption Maximization}]
Given a directed graph $G=(V,E)$, a state $\mathcal{S}$, and a budget $k$, compute
$\mathcal{F}_\mathcal{S}^{*}:=\max_{|A|\le k, A\cap \mathcal{S}^b=\emptyset}\mathcal{F}_\mathcal{S}(A)$.
\end{definition}
% \vspace{-0.2cm}

We impose the restriction that $A\cap \mathcal{S}^b=\emptyset$, that is, the customers of the blue product cannot be targeted. However, it is straightforward to see all our results also hold when this restriction is relaxed.

\textbf{Submodularity and Monotonicity.} Consider an arbitrary function $f(\cdot)$ which maps subsets of a ground set $V$ to non-negative real values. Function $f$ is \textit{submodular} if for all $v\in V$ and $A\subseteq A'\subseteq V$, $f$ satisfies $f(A'\cup\{v\})-f(A')\le f(A\cup\{v\})-f(A)$.
% the marginal gain from adding an elemnt to a set $A$ is at least as large as the marginal gain from adding the same element to a superset of $A$. More precisely, 
Furthermore, we say that $f$ is \textit{monotonically increasing} if $f(A\cup\{v\})\ge f(A)$.

\textbf{Approximation Algorithms.} We say that $\mathcal{A}$ is a $\rho$-approximation algorithm for a maximization problem $P$ and some $\rho>0$ if the output of $\mathcal{A}$ is not smaller than the optimal solution times $\rho$ for any instance of problem $P$.

\textbf{Some Inequalities.} We utilize some standard probabilistic inequalities~\cite{dubhashi2009concentration} such as Chernoff bound, Markov's inequality, and Chebyshev’s inequality, which are provided in Appendix~\ref{appendix:ineq}. Furthermore, we sometimes use the basic inequalities $1-z\le \exp(-z)$ for any $z$ and $4^{-z}\le 1-z$ for any $0<z<1/2$. 

\textbf{Assumptions.}
% All logarithms are to base $e$, otherwise we point out explicitly.
We let $n$ (the number of nodes) tend to infinity. We say an event $\mathcal{E}$ occurs with high probability (w.h.p.) if it happens with probability $1-o(1)$.

\subsection{Our Contribution}
% Devising effective strategies to accomplish successful marketing in the \textsc{Product Adoption Maximization} problem is very crucial. (See Appendix~\ref{appendix:example} for a concrete example.)
We prove that there is no $(1-\frac{1}{e}+\epsilon)$-approximation algorithm (for any constant $\epsilon>0$) for the \textsc{Product Adoption Maximization} problem, under some plausible complexity assumptions, by a reduction from the \textsc{Maximum Coverage} problem~\cite{feige1998threshold}.

We show that the objective function $\mathcal{F}_{S}(\cdot)$ is monotone and submodular. Given a pick profile, we build an extended sequence $es(v)$ for each node $v$. This sequence has the distinctive characteristic that the final color adopted by $v$ is the same as the color of the first node in this sequence which is red/blue in the initial state $\mathcal{S}_0$ (if such a node does not exist, $v$ will be uncolored). This will be the main building block of our proof of monotonicity and submodularity. Consequently, a greedy \textsc{Hill Climbing} approach would provide us with an approximation ratio of $1-\frac{1}{e}$, cf.~\cite{nemhauser1978analysis}. However, for that we need to repeatedly compute $\mathcal{F}_{\mathcal{S}}(\cdot)$, which seems to be computationally expensive. Thus, we resort to the Monte Carlo approximation for estimating $\mathcal{F}_{\mathcal{S}}(\cdot)$, which overall provides us with a polynomial time $(1-\frac{1}{e}-\epsilon$)-approximation algorithm for any $\epsilon>0$. Moreover, our empirical evaluations on a large spectrum of real-world and synthetic networks illustrate that our proposed algorithm consistently outperforms the classic centrality-based algorithms. Therefore, our algorithm not only boasts proven theoretical guarantees, but also demonstrates highly satisfactory performance in practice.

In the second part, we provide several tight bounds on the convergence time of the \textsc{Random Pick} process in terms of different graph parameters. This is not only very fundamental and interesting by its own sake, but also a prerequisite for bounding the time complexity of our algorithm. Specifically, we prove the upper bounds of $\tilde{\mathcal{O}}(m)$ and $\tilde{\mathcal{O}}(D\Delta_+)$, where $\tilde{\mathcal{O}}$ is used to hide poly-logarithmic terms in $n$. We prove that these bounds are the best possible and further derive stronger bounds for undirected graphs. To prove the bounds, we introduce the novel concept of traversed node chain and rely on various Markov chain analyses. For the tightness proofs, we present explicit constructions. We also consider the randomized setup, where each node is colored independently with probability (w.p.) $q>0$ and prove the bound of $\tilde{\mathcal{O}}(1/q^2)$ for directed graphs and $\tilde{\mathcal{O}}(1/q)$ for undirected graphs. Additionally, we investigate the convergence time on various real-world and synthetic networks.

Finally, we prove that the problem of determining whether there exists an initial state with $k$ colored nodes with the expected convergence time $t$, 
%for given integers $k,t$ and graph $G$, 
is NP-hard (even in very restricted setups) by a reduction from the \textsc{Vertex Cover} problem.

% \vspace{-0.3cm}
\subsection{Related Work}

\textbf{IC Model.}
The Independent Cascade (IC) model, popularized by the seminal work of Kempe et al.~\cite{kempe2003maximizing}, has obtained substantial popularity to simulate viral marketing, cf.~\cite{li2018influence, schoenebeck2020limitations}. In this model, initially each node is uncolored (inactive), except a set of seed nodes which are colored (active). Once a node is colored, it gets one chance to color each of its out-neighbors. The problem of finding a seed set of size $k$ which maximizes the expected final number of colored nodes have been studied extensively and a large collection of approximation and heuristic algorithms (mostly using greedy approaches) have been developed, cf.~\cite{kempe2003maximizing,li2018influence,chen2010scalable}. Different extensions of the IC model have been introduced to investigate the diffusion of multiple competing products, cf.~\cite{lin2015analyzing, myers2012clash, bharathi2007competitive, lu2015competition, carnes2007maximizing, zhu2016minimum, datta2010viral}. They usually suppose that both products (red or blue) spread following the IC model and define how the spread of one can influence the other.
% each colored node could be either red or blue, then once an uncolored node is colored by a red (resp. blue) node, it becomes red (resp. blue), and ties (a red and a blue node attempt to color an uncolored at once) are broken in a predefined manner.
Similar to our work, their main goal is to design efficient algorithms for the selection of a set of red seed nodes. The problem is proven to be NP-hard in most scenarios and thus the previous works have resorted to approximation algorithms for general case~\cite{lu2015competition, carnes2007maximizing, datta2010viral} or exact algorithm for special cases~\cite{bharathi2007competitive}.

\noindent \textbf{Threshold Model.}
In the Threshold model, each node $v$ has a threshold $\tau(v)$. From a starting state, where each node is either colored or uncolored, an uncolored node becomes colored once $\tau(v)$ fraction of its out-neighbors are colored. The problem of selecting $k$ colored seed nodes cannot be approximated within the ratio of $O(2^{\log^{1-\epsilon}n})$, for any constant $\epsilon>0$, unless $NP\subseteq DTIME(n^{polylog(n)})$~\cite{chen2009approximability}. However, the problem is traceable for trees~\cite{centeno2011irreversible} and there is a $(1-1/e)$-approximation algorithm for the Linear Threshold (LT) model, where the threshold $\tau(v)$ is chosen uniformly and independently at random in $[0,1]$, cf.~\cite{kempe2003maximizing}. Several works~\cite{liu2016containment, pathak2010generalized, wu2015maximizing, borodin2010threshold} have considered the setup with two colors, where for a node $v$, once $\tau(v)$ fraction of its out-neighbors are colored, it picks one of the two colors following a certain updating rule. Again due to the NP-hard nature of the problem, approximation techniques using submodularity~\cite{liu2016containment} and rapidly mixing Markov chains~\cite{pathak2010generalized} and various heuristics~\cite{wu2015maximizing} have been developed.

\noindent \textbf{Majority-based Model.}
Let each node be either red or blue. Then, in the Majority model~\cite{chistikov2020convergence,zhuang2020lifting, zehmakan2020opinion,zehmakan2021majority,ref1}, in every round each node updates its color to the most frequent color in its out-neighborhood and in the Voter model~\cite{hassin2001distributed, richardson2002mining, cooper2014power} each node chooses one of its out-neighbors at random and adopts its color. Unlike the IC or Threshold model, here a node can switch back and forth between red and blue. These models aim to mimic opinion formation process (where individuals might change their opinions constantly) while the IC and Threshold model goal to simulate product adoption process (where once a customer adopts a product, it is costly to switch to the other). For the problem of maximizing the final number of red nodes using a budget $k$, a $\log \Delta$-approximation algorithm is known for the Majority model~\cite{mishra2002hardness} and a Fully Polynomial Time Approximation Scheme~\cite{even2007note} for the Voter model (when selecting each seed node $v$ has a given cost $c(v)$). Exact algorithms for special graph classes~\cite{centeno2011irreversible} and heuristic centrality-based algorithms~\cite{fazli2014non} have been investigated too.
% An extension of the Majority model, where some nodes are uncolored was proposed in~\cite{zehmakan2021majority}. They characterize classes of graphs which are resilient against a spectrum of adversarial attacks.

\noindent \textbf{Convergence Time.}
The convergence time is one of the most well-studied characteristic of dynamic processes, cf.~\cite{auletta2019consensus, auletta2018reasoning,n2020rumor}. For the Majority model on undirected graphs, it is proven~\cite{poljak1986pre} that the process converges in $\mathcal{O}(n^2)$ rounds (which is tight up to some poly-logarithmic factor~\cite{frischknecht2013convergence}). Better bounds are known for special graphs~\cite{zehmakan2020opinion}. The convergence properties have also been studied for directed acyclic graphs~\cite{chistikov2020convergence}, weighted graphs~\cite{keller2014even}, and when the updating rule is biased~\cite{lesfari2022biased}. For the Voter model, an upper bound of $\mathcal{O}(n^3\log n)$ has been proven in~\cite{hassin2001distributed} using reversible Markov chain argument. For the Threshold model on undirected graphs, it is proven~\cite{zehmakan2019spread} that when $\tau(v)=r/d_+(v)$ for a fixed $r$, then the convergence time can be bounded by $\mathcal{O}(rn/\delta)$, where $\delta$ is the minimum degree. In the Push-Pull protocol on undirected graphs, each node is either informed or uninformed. Then, every informed node selects a random neighbor to inform (push) and each uninformed node selects a random neighbor to check whether it has the information (pull). The process stops when no new node could become informed. The convergence time of the process has been studied extensively on undirected graphs, and it is shown that the expansion properties of the underlying graph are the governing parameters, cf.~\cite{giakkoupis2011tight}.

% \vspace{-0.1cm}
\section{Product Adoption Maximization}
\label{algo-sec}

\subsection{Inapproximability Result}

\begin{theorem}
\label{hardness}
There is no polynomial time $(1-\frac{1}{e}+\epsilon)$-approximation algorithm (for any constant $\epsilon>0$) for the \textsc{Product Adoption Maximization} problem, unless $NP \subseteq DTIME(n^{\mathcal{O}(\log \log n)})$.
\end{theorem}

\textit{Proof Sketch.} The reduction is from the \textsc{Maximum Coverage} problem, where for a given collection of subsets $S=\{S_1, S_2,\cdots, S_l\}$ of an element set $O=\{O_1,\cdots, O_h\}$ and an integer $k$, the goal is to find the maximum number of elements covered by $k$ subsets. It is known that there is no polynomial time $(1-\frac{1}{e})$-approximation algorithm for the \textsc{Maximum Coverage} problem, unless $NP \subseteq DTIME(n^{\mathcal{O}(\log \log n)})$, cf.~\cite{feige1998threshold}. Given an instance of \textsc{Maximum Coverage} problem, we can construct an instance of the \textsc{Product Adoption Maximization} problem in polynomial time such that $OPT_{MC}=(OPT_{PAM}-k)/\lceil 1/\epsilon \rceil$, where $OPT_{MC}$ and $OPT_{PAM}$ correspond to the optimal solution in the \textsc{Maximum Coverage} and \textsc{Product Adoption Maximization} problem. Combining this with the aforementioned hardness result for \textsc{Maximum Coverage} concludes the proof. A complete proof is given in Appendix~\ref{appendix:hardness}. \hfill$\square$

% \vspace{-0.3cm}
\subsection{Greedy Algorithm}
\label{sec:greedy}

Our first goal here is to prove Theorem~\ref{monotone-submod}, which states that the function $\mathcal{F}_{\mathcal{S}}(\cdot)$ is monotone and submodular. Let us present a reformulation of our model, which facilitates the proof of the theorem.

\textbf{Extended Sequence.} Consider a pick profile $\mathcal{P}$ (defined in Section~\ref{preliminaries-sec}), then we establish the notion of \emph{extended sequence} $es(v)$ for a node $v$. Let us define $es^t(v)$ for $t\ge 0$ in a constructive manner. Define $es^0(v)=v$ and for $t\ge 1$, $es^t(v)=es^{t-1}(v), es^{t-1}(v')$, where $v'$ is the node $v$ picks in the $t$-th round, i.e., $v'$ is the $t$-th element in the pick sequence $ps(v)$. (We use ``,'' for concatenation.) Then, the extended sequence $es(v)$ is just $es^t(v)$ when we let $t$ go to infinity. Note that this definition is well-defined since we can first calculate $es^0$ for all nodes, then $es^1$, and so on. However, the time complexity of computing these sequences does not concern us since they are solely utilized to provide a reformulation of our model, which simplifies the proof of submodularity and monotonicity.

The functionality of the extended sequence $es(v)$ is that, based on Lemma~\ref{extended-seq-lemma}, if we keep traversing its nodes until we reach a node which is colored red/blue in the initial state $\mathcal{S}_0$, then that would be the final color picked by $v$. The idea is to trace back where the color of $v$ comes from. Please refer to Appendix~\ref{appendix:es-ex} for an example.

\begin{lemma}
\label{extended-seq-lemma}
Consider the \textsc{Random Pick} process on a graph $G=(V,E)$ with an initial state $\mathcal{S}_0$ and pick profile $\mathcal{P}$. For a node $v$, let $w$ be the first node in $es^{t}(v)$ such that $\mathcal{S}_0(w)\ne u$, then $\mathcal{S}_{t}(v)=\mathcal{S}_0(w)$, and if there is no such node then $\mathcal{S}_t(v)=u$.
\end{lemma}

The proof of Lemma~\ref{extended-seq-lemma} follows from an induction on $t$. A proof is provided in Appendix~\ref{appendix:lemma-es} for the sake of completeness.

\begin{theorem}
\label{monotone-submod}
For an arbitrary state $\mathcal{S}$ on a graph $G=(V,E)$, the function $\mathcal{F}_{\mathcal{S}}(\cdot)$ is monotonically increasing and submodular.
\end{theorem}
\begin{proof}
\textit{Submodularity.} Consider a node $v\in V$ and subsets $A\subseteq A'\subseteq \mathcal{S}^u$. We want to prove that $\mathcal{F}_{\mathcal{S}}\left( A'\cup \{v\}\right) -\mathcal{F}_{\mathcal{S}}\left( A'\right)\le \mathcal{F}_{\mathcal{S}}\left( A\cup \{v\}\right) -\mathcal{F}_{\mathcal{S}}\left( A\right)$. Note that for $\mathcal{F}_{\mathcal{S}}\left( \cdot\right)$, we need to compute the expected final number of red nodes, which seems difficult to deal with directly. However, if we fix all the random choices of out-neighbors, then it perhaps becomes easier to handle. Let $\mathcal{F}_{\mathcal{S}}^{\mathcal{P}}\left( \cdot\right)$ be the same as $\mathcal{F}_{\mathcal{S}}\left( \cdot\right)$ but conditioning on the pick profile $\mathcal{P}$. Then, we can rewrite $\mathcal{F}_{\mathcal{S}}\left( \cdot\right)$ as $\sum_{\textrm{all profiles } \mathcal{P}} \Pr[\mathcal{P}]\cdot \mathcal{F}_{\mathcal{S}}^{\mathcal{P}}\left( \cdot\right)$.
Since a non-negative linear combination of submodular functions is also submodular, it suffices to prove that $\mathcal{F}_{\mathcal{S}}^{\mathcal{P}}\left( \cdot\right)$ is submodular.

Consider a node $w$ which counts as a final red node in $\mathcal{F}_{\mathcal{S}}^{\mathcal{P}}(A'\cup \{v\})$ but not in $\mathcal{F}_{\mathcal{S}}^{\mathcal{P}}(A')$. According to Lemma~\ref{extended-seq-lemma}, this implies that in the extended sequence $es(w)$ (with respect to $\mathcal{P}$): (i) there is a node from $\mathcal{S}^b$ before any node from $A'\cup \mathcal{S}^r$ or there is no node from $A'\cup\mathcal{S}^b\cup \mathcal{S}^r$, (ii) node $v$ appears before any node in $A'\cup \mathcal{S}^b\cup \mathcal{S}^r$. (Note that $A'\cup \mathcal{S}^b\cup \mathcal{S}^r$ is the set of colored nodes.) Condition (i) is true since $w$ does not count as a final red node in $\mathcal{F}_{\mathcal{S}}^{\mathcal{P}}(A')$ and condition (ii) holds since it does in $\mathcal{F}_{\mathcal{S}}^{\mathcal{P}}(A'\cup \{v\})$. Note that since $A\subseteq A'$, both conditions will remain true if we replace $A'$ with $A$. This implies that $w$ counts as a final red node in $\mathcal{F}_{\mathcal{S}}^{\mathcal{P}}(A\cup \{v\})$ but not in $\mathcal{F}_{\mathcal{S}}^{\mathcal{P}}(A)$. Therefore, we can conclude $\mathcal{F}_{\mathcal{S}}^{\mathcal{P}}\left( A'\cup \{v\}\right) -\mathcal{F}_{\mathcal{S}}^{\mathcal{P}}\left( A'\right)\le \mathcal{F}_{\mathcal{S}}^{\mathcal{P}}\left( A\cup \{v\}\right) -\mathcal{F}_{\mathcal{S}}^{\mathcal{P}}\left( A\right)$.

\textit{Monotonicity.} Consider an arbitrary node $v\in V$ and set $A\subset V$. We aim to show that $\mathcal{F}_{\mathcal{S}}\left(A\right)\le \mathcal{F}_{\mathcal{S}}\left(A\cup \{v\}\right)$. Again using $\mathcal{F}_{\mathcal{S}}\left( \cdot\right)=\sum_{\textrm{all profiles } \mathcal{P}} \Pr[\mathcal{P}]\cdot \mathcal{F}_{\mathcal{S}}^{\mathcal{P}}\left( \cdot\right)$, it suffices to prove that $\mathcal{F}_{\mathcal{S}}^{\mathcal{P}}\left(A\right)\le \mathcal{F}_{\mathcal{S}}^{\mathcal{P}}\left(A\cup \{v\}\right)$. Let node $w$ count as a final red node in $\mathcal{F}_{\mathcal{S}}^{\mathcal{P}}\left(A\right)$. According to Lemma~\ref{extended-seq-lemma}, in the extended sequence $es(w)$ (with respect to $\mathcal{P}$), there is a node from $\mathcal{S}^r\cup A$ which appears before any node from $\mathcal{S}^b$. This condition obviously remains true if we replace $A$ with $A\cup \{v\}$. Therefore, $w$ counts as a final red node in $\mathcal{F}_{\mathcal{S}}^{\mathcal{P}}\left(A\cup \{v\}\right)$ too.
\end{proof}

The greedy \textsc{Hill Climbing} algorithm for the \textsc{Product Adoption Maximization} problem, repeatedly finds the element $v\in (V\setminus \mathcal{S}^b)\setminus A$ which maximizes $\mathcal{F}_{\mathcal{S}}(A\cup \{v\})$ and adds it to $A$ until $|A|=k$. (This is essentially the same as Algorithm~\ref{greedy} if we replace $Est_{\mathcal{S}}(A\cup\{w\})$ with $\mathcal{F}_{\mathcal{S}}(A\cup \{w\})$.) Since $\mathcal{F}_{\mathcal{S}}(\cdot)$ is non-negative, monotone, and submodular (see Theorem~\ref{monotone-submod}), the \textsc{Hill Climbing} algorithm has the approximation ratio of $(1-\frac{1}{e})$, according to~\cite{nemhauser1978analysis}. However, since computing $\mathcal{F}_{\mathcal{S}}(\cdot)$ seems to be computationally expensive, we use the Monte Carlo method, which results in Algorithm~\ref{greedy} whose analysis is given in Theorem~\ref{analysis-algo-thm}.

\begin{small}
\begin{algorithm}[h]
\caption{Monte Carlo Greedy Algorithm}
\label{greedy}
 \KwIn{Graph~$G=(V,E)$,~state~$\mathcal{S}$,~budget~$k$,~error~factor~$\epsilon>0$.}
\KwOut{Selected seed nodes.}
% \vspace{0.2cm}
Initialize $A\leftarrow \emptyset$

\For{$i=1$ {\normalfont to} $k$}{
$v\leftarrow \argmax_{w\in (V\setminus \mathcal{S}^b)\setminus A} Est_{\mathcal{S}}(A\cup\{w\})$
        
$A\leftarrow A\cup \{v\}$
     
}

\Return $A$

\vspace{0.2cm}

\SetKwFunction{FEs}{$Est_{\mathcal{S}}$}
    
\SetKwProg{Fn}{Function}{:}{}
    
\Fn{\FEs{$A$}}{
$count \leftarrow 0$

\For{$j=1$ {\normalfont to} $R=(27nk^2\ln n^3)/\epsilon^2$}{
Simulate \textsc{Random Pick} with the state obtained from $\mathcal{S}$ by making $A$ red.
        
$r \leftarrow \textrm{final number of red nodes}$

$count \leftarrow count + r$
        
}
\Return $count/R$
}
\textbf{end}
\end{algorithm}
% \vspace{-1cm}
\end{small}

\begin{theorem}
\label{analysis-algo-thm}
Algorithm~\ref{greedy} w.h.p. achieves approximation ratio of $1-\frac{1}{e}-\epsilon$ in time $\mathcal{O}(kRnmT)$, where $T=D\Delta_+$.
\end{theorem}
% \vspace{-1cm}
\begin{proof}
In the context of \textsc{Hill Climbing} approach, it is known (e.g., Theorem 3.6 in~\cite{chen2013information}) that if $Est_{\mathcal{S}}(\cdot)$ is a multiplicative $\eta$-error estimate of $\mathcal{F}_{\mathcal{S}}(\cdot)$ for $\eta= \epsilon/(3k)$, then Algorithm~\ref{greedy} has an approximation ratio of $1-\frac{1}{e}-\epsilon$.

Consider an arbitrary seed set $A$. Let random variable $x_j$, for $1\le j\le R$, be the number of red nodes in the $j$-th simulation (i.e., the $j$-th iteration of \textbf{for} loop in line 9 of Algorithm~\ref{greedy}) divided by $n$. Define $X:=\sum_{j=1}^{R}x_j$ and $\bar{X}:=X/R$. Applying the Chernoff bound (see Appendix~\ref{appendix:ineq}), we have $\Pr[|\bar{X}-\mathbb{E}[\bar{X}]|\ge \eta \mathbb{E}[\bar{X}]]  =  \Pr[|X-\mathbb{E}[X]|  \ge \eta \mathbb{E}[X]] \le   2\exp(-\eta^2\mathbb{E}[X]/3)\le 2/n^3$. In the last step, we used $\eta^2=\epsilon^2/(9k^2)$, $\mathbb{E}[X]\ge R/n$ (since $x_i\ge 1/n$) and $R=(27nk^2\ln n^3)/\epsilon^2$.

Since there are at most $nk\le n^2$ calls to $Est_{\mathcal{S}}(\cdot)$, a union bound implies that $Est_{\mathcal{S}}(\cdot)$ is a multiplicative $\eta$-error estimate of $\mathcal{F}_{\mathcal{S}}(\cdot)$ w.h.p.

The algorithm performs at most $knR$ simulations of the \textsc{Random Pick} process. The time to execute one round of a simulation is in $\mathcal{O}(m)$. Furthermore, in Section~\ref{sec:conv} (Theorems~\ref{diam-time-thm}) we prove that w.h.p. the convergence time of the process is in $\mathcal{O}(D\Delta_+)$. Therefore, the overall run time of the algorithm is in $\mathcal{O}(kRnmT)$ w.h.p. 
\end{proof}

\textbf{Run Time on Real-world Networks.} According to Theorem~\ref{analysis-algo-thm}, the run time of Algorithm~\ref{greedy} in $\tilde{\mathcal{O}}(k^3n^2mD\Delta_+/\epsilon^2)$. In real-world social networks, it is commonly observed that $m=\Theta(n)$ and $D$ and $\Delta_+$ are small, cf.~\cite{albert2002statistical}. Thus, for fixed values of $k$ and $\epsilon$, the algorithm runs in $\tilde{\mathcal{O}}(n^3)$. We pose devising heuristic algorithms which run in nearly linear time as a potential avenue for future research in Section~\ref{sec:conclusion}.

\textbf{Extensions.} The \textsc{Random Pick} process can be naturally extended to encompass the case of more than two colors, where again an uncolored node picks an out-neighbor at random and adopts its color. In the \textsc{Product Adoption Maximization} problem, we can let the initial colored nodes be of different colors (rather than just blue) and we again simply add red seed nodes. Then, all of our results can be easily extended to this setup.

Furthermore, we can extend the greedy algorithm to the case where the costs of selected nodes are non-uniform and the total cost of selected nodes cannot exceed a given budget. The adapted greedy algorithm again achieves an approximation ratio of $1-\frac{1}{e}-\epsilon$, cf.~\cite{sviridenko2004note, krause2005note}.

%%%%%%%%%%%%%%%%%%%%%%%%%	
\begin{figure}[t]
		\centering
		\includegraphics[width=0.9\linewidth]{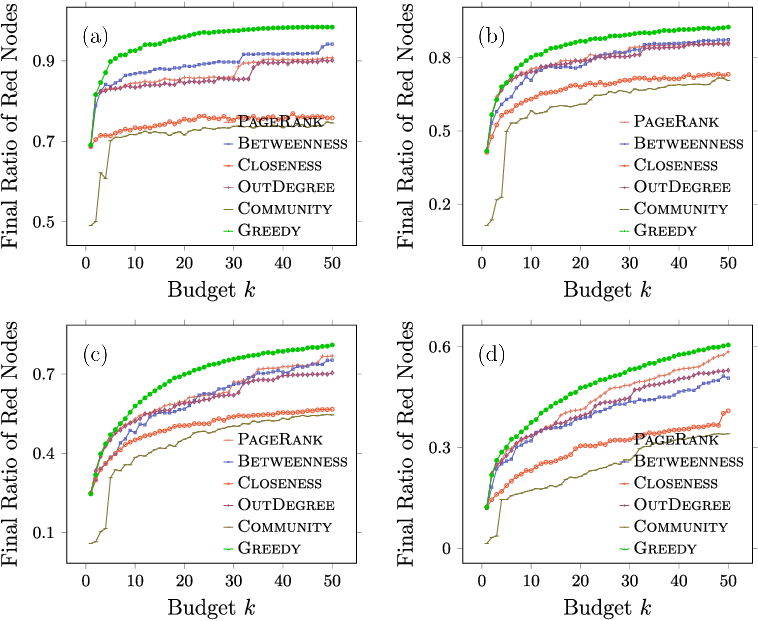}
  \vspace{-0.2cm}
		\caption{The final ratio of red nodes for $k$ seed nodes and $b_0=$ (a) 10, (b) 20, (c) 50 and (d) 100 initial blue nodes on the Food network. \label{fig:2}}
  % \vspace{-0.2cm}
\end{figure}
%%%%%%%%%%%%%%%%%%%%%%%%%%%

\subsection{Experiments: Comparison of Algorithms}
\label{sec:exp-alg}

\textbf{Datasets.} Our experiments are conducted on real-world social networks from KONECT~\cite{Ku13}, Network Repository~\cite{RoAh15}, and SNAP~\cite{LeSo16}. We also consider the BA (Barabási–Albert) model~\cite{albert2002statistical} which is a synthetic graph model tailored to capture fundamental properties observed in real-world networks, such as small diameter and scale-free degree distribution. For the BA graphs, we set the parameters such that the average degree is comparable to that of experimented real-world networks with similar number of nodes. The list of experimented networks is in Table~\ref{tab:conv_1}. We also provide a table including some graph parameters for these networks in Appendix~\ref{appendix:table}.

\textbf{Machine.} All our experiments, programmed in \textit{Julia}, were conducted on a Linux server with 128G RAM and 4.2 GHz Intel i7-7700 CPU, using a single thread.

% For the purpose of replication, we make  the source code  publicly available on \url{https://github.com....}. \ahad{We might submit to a conference with anonymous review process. Thus, cannot include that until acceptance.}

\textbf{Algorithms.} We compare our algorithm against several classic centrality based algorithms which choose $k$ uncolored nodes to be the seed nodes according to some centrality measure for an instance of the \textsc{Product Adoption Maximization} problem: \textsc{PageRank}~\cite{page1999pagerank} (highest PageRank), \textsc{Closeness}~\cite{Be65,BrPi07} (highest closeness), {Betweenness}~\cite{Br01,BrPi07} (highest betweenness), \textsc{InDegree}~\cite{Be65,BrPi07}: (highest in-degree), \textsc{OutDegree}~\cite{Be65,BrPi07} (highest out-degree). (For undirected graphs, we only consider the \textsc{OutDegree} algorithm because \textsc{InDegree} and \textsc{OutDegree} are identical.)

Furthermore, we consider a \textsc{Community} based algorithm where we run a community detection algorithm~\cite{RaAlKu07} to find at least $2k$ communities. Then, we sort the communities based on the number of blue nodes in them in the ascending order. In each of the first $k$ communities, we select an uncolored node at random and color it red. (The idea is that the red marketer targets communities which are ``untouched'' or ``less touched'' by the blue marketer.) We set $R=300$ for our greedy algorithm in all the experiments.

% \textsc{PageRank}, \textsc{Betweenness}, \textsc{Closeness}, \textsc{InDegree}, and \textsc{OutDegree}, and the community detection based baseline, \textsc{Community}. The baseline \textsc{Community} is conducted as follows:
% (1) We run a community detection algorithm~\cite{RaAlKu07} such that it outputs as many communities as possible;
% (2) For $k$ rounds, we pick one community with the least number of colored nodes and  randomly choose one uncolored node in this community.
% \textbf{Fast greedy algorithm.}
% \textcolor{red}{Unlike the algorithm \textsc{Greedy}, algorithm \textsc{FastGreedy} approximates the marginal gain of nodes with top $1\%$ highest degrees and selects one node with the largest marginal gain in each iteration, instead of approximating the marginal gain for each uncolored node.   (Is \textsc{FastGreedy} a baseline or our proposed algorithm, as part of our contribution?)} \ahad{I was thinking of as part of our contribution, but since it actually performs poorly, can you please remove it from the diagrams? I will leave finding a faster algorithm as an open problem in the appendix.}

\textbf{Comparison.} For each network, we randomly selected $b_0=10,20,50,100$ nodes and colored them blue. Then, we chose $k = 1,2,\ldots,50$ initial red nodes from uncolored nodes using each of the aforementioned algorithms. We ran the process and computed the final number of red nodes. We repeated this process 300 times and computed the average final ratio of red nodes for each algorithm. The results for Food network are depicted in Figure~\ref{fig:2}. Similar diagrams are given in Appendix~\ref{appendix:alg-diagram} for other networks. We also provide a table in Appendix~\ref{appendix:alg-diagram}, which reports the average final ratio of red nodes for each algorithm, including the standard deviations.

  % \vspace{-0.2cm}
We observe that our proposed algorithm consistently outperforms other algorithms. Therefore, it not only possesses proven theoretical guarantees, but also performs very well in practice.

Another interesting observation is that the outcomes of our experiments emphasize on the importance of following a clever marketing strategy. For example, in Figure~\ref{fig:2} (a), 10 blue nodes are chosen randomly. If the red marketer follows our greedy approach, even with 1 red seed node it can win almost 70\% of the customers at the end and with $10$ red seed nodes, it wins over 90\% of the whole network. A similar behavior can be observed in other setups.

% final ratio of red nodes when $k=50$ for different networks and the different numbers of initial blue nodes in Table~\ref{tab:time}. 
% Numerical results show that Algorithm \textsc{Greedy} outperforms other baselines.
%We also record the running time of two greedy algorithms in Table~\ref{tab:time}. Numerical results show that Algorithm \textsc{Greedy} outperforms other baselines.

% Effectiveness(baselines)(figure)
% Efficiency(fastgreedy, greedy)(time)
% Standard deviation(table)
% \vspace{-0.3cm}
\section{Convergence Time}
\label{sec:conv}

\subsection{Tight Upper Bounds}
In the \textsc{Random Pick} process on a graph $G=(V,E)$ with the initial state $\mathcal{S}_0$, if an uncolored node does not have a path to (i.e., cannot reach) any node in $R_0\cup B_0$, it remains uncolored forever, and otherwise, it eventually becomes colored (i.e., red/blue). To bound the convergence time of the process, it suffices to bound the number of rounds required for all such nodes to be colored. (Note that we only focus on whether a node is colored or uncolored since the exact color of a node (red/blue) does not affect the consensus time.) Before proving our bounds and their tightness, we provide some preliminaries

\begin{observation}
\label{exp-pick-obs}
In the \textsc{Random Pick} process on a graph $G=(V,E)$, the expected number of rounds a node $v$ needs to pick a specific out-neighbor $w\in \Gamma_{+}(v)$ is equal to $1/d_+(v)$.
\end{observation}

\textbf{Traversed Node Chain.} Consider a node chain $w_0, w_1, \cdots, w_h$ in a graph $G=(V,E)$ where $(w_i,w_{i-1})\in E$ for each $1\le i\le h$. In the \textsc{Random Pick} process on $G$, we say $w_0, w_1, \cdots, w_h$ is \textit{traversed} if there is a sequence of rounds such that $w_1$ picks $w_0$, then $w_2$ picks $w_1$, and so on until $w_{h}$ picks $w_{h-1}$. More precisely, there is $1\le t_1< t_2< \cdots< t_{h}$ such that $ps_{t_{i}}(w_i)=w_{i-1}$ for $1\le i\le h$. Now, we state Lemma~\ref{trav-chain-lemma}, which is a main building block of our proofs in this section.

\begin{lemma}
\label{trav-chain-lemma}
Consider the \textsc{Random Pick} process on a graph $G=(V,E)$. The expected number of rounds for a node chain $w_0,w_1, \cdots, w_h$ to be traversed is equal to $\sum_{i=1}^{h}d_+(w_i)$.
\end{lemma}
\begin{proof}
For any $1\le i\le h$, the expected number of rounds for $w_i$ to pick node $w_{i-1}$ is $d_+(w_i)$. This is because $w_i$ has $d_+(w_i)$ out-neighbors, and it picks one of them uniformly and independently at random in each round. Thus, the expected number of rounds to traverse the node chain $w_0,w_1, \cdots, w_h$ is equal to $\sum_{i=1}^{h}d_+(w_i)$.
\end{proof}

\begin{theorem}
\label{diam-time-thm}
The convergence time of the \textsc{Random Pick} process w.h.p. is at most $4D\Delta_+\log_2 n$ if $G$ is directed and $20n\log_2 n$ if $G$ is undirected.
\end{theorem}
\begin{proof}
\textit{Directed.} Consider an arbitrary initial state $\mathcal{S}_0$ on $G$. Let $w$ be a node which has a path to some node $w_0\in R_0\cup B_0$. Consider the node chain $\textbf{w}=w_0,w_1,\cdots, w_h$ corresponding to the shortest path from $w=w_h$ to $w_0$, where $h=d(w,w_0)$. Let $T_{w}$ be the number of rounds required for $\textbf{w}$ to be traversed. Applying Lemma~\ref{trav-chain-lemma} yields $\mathbb{E}[T_w]\le \sum_{i=1}^{h}d_+(w_i)\le h\Delta_+\le D\Delta_+$, where we used $d_+(w_i)\le \Delta_+$ and $h\le D$ which are correct by definition. Using Markov's inequality (see Appendix~\ref{appendix:ineq}), we have $\Pr[T_w\ge 2D\Delta_+]\le \frac{\mathbb{E}[T_w]}{2D\Delta_+}\le \frac{1}{2}$.

Let $2D\Delta_+$ consecutive rounds be called a \textit{phase}. Let us partition the rounds of the process into phases: rounds $1$ to $2D\Delta_+$ build phase 1, rounds $2D\Delta_++1$ to $4D\Delta_+$ build phase 2 and so on. The probability that $\textbf{w}$ is not traversed in any of the first $2\log_2 n$ phases is at most $(1/2)^{2\log_2 n}=1/n^2$. Thus, after $2\log_2 n$ phases (i.e., $4D\Delta_+\log_2 n$ rounds) w.p. at least $1-1/n^2$, the node chain $\textbf{w}$ is traversed.
By a simple inductive argument, if $\textbf{w}$ is traversed, then node $w$ is colored. Thus, node $w$ will be colored after $4D\Delta_+\log_2 n$ rounds w.p. $1-1/n^2$. Since there are at most $n$ nodes which will be colored eventually (i.e., the uncolored nodes which have a path to $R_0\cup B_0$), a union bound implies that w.p. $1-1/n$ in $4D\Delta_+\log_2 n$ rounds all such nodes are colored, and the process has converged.

\textit{Undirected.} It suffices to prove that if $G$ is undirected, then $\sum_{i=1}^{h}d_+(w_i)\le 5n$ (instead of $\sum_{i=1}^{h}d_+(w_i)\le D\Delta_+$). Then, the rest of the proof is identical to the one for the directed case by replacing $D\Delta_+$ with $5n$.

Consider a node $v$ which is not among the nodes in the node chain $\textbf{w}$. Node $v$ can be adjacent to two nodes $w_i$ and $w_j$ only if they are at most one apart on the node chain (e.g., $v$ cannot be adjacent to both $w_{i}$ and $w_{i+3}$) because otherwise there is a path from $w_h$ to $w_0$ whose length is smaller than $h$, which is in contradiction with $\textbf{w}$ corresponding to a shortest path from $w_h$ to $w_0$. (We are using the fact that if $v$ has an edge to $w_i$, then $w_i$ has an edge to $v$ too since $G$ is undirected.) Hence, node $v$ is adjacent to at most $3$ nodes on $\textbf{w}$. If we look at all the edges which count against the sum $\sum_{i=1}^{h}d_+(w_i)$, they are either between nodes in $\textbf{w}$ or between a node in $\textbf{w}$ and a node outside it. There are at most $2n$ edges between nodes in $\textbf{w}$ (we are using the point that there is no edge between $w_i$ and $w_{i-2}$, again because $\textbf{w}$ corresponds to a shortest path). Furthermore, the number of edges of the second type is upper-bounded by $3n$ according to our observation from above. Thus, we get $\sum_{i=1}^{h}d_+(w_i)\le 5n$.
\end{proof}

\textbf{Tightness.} Consider a node set $V_1$ of size $n/2$, a node set $V_2$ of size $n/2-1$, and a node $v$. Add an edge from every node in $V_1$ to every node in $V_2\cup\{v\}$ to construct graph $G$ and assume that initially only node $v$ is colored. Note that $\Delta_+(G)=n/2$ and $D(G)=1$. Theorem~\ref{diam-time-thm} states that the convergence time is at most $4\Delta_+D\log_2 n=2n\log_2 n$ w.h.p. We show that this bound is tight up to a constant factor. More precisely, we prove that if we replace $2n\log_2 n$ with $t^*=(1/8)n\log_2 n$, the statement is no longer true.

Let us label nodes in $V_1$ from $v_1$ to $v_{n/2}$. Let Bernoulli random variable $x_i$, for $1\le i\le n/2$, be 1 if and only if $v_i$ is uncolored in round $t^*$. Node $v_i$ becomes colored only if it picks $v$; thus, $\mathbb{E}[x_i]=\Pr[x_i=1]\ge (1-2/n)^{t^*}\ge (1/4)^{2t^*/n}=(1/2)^{\log_2 n^{1/2}}=1/\sqrt{n}$, where we used the estimate $1-z\ge (1/4)^z$ for $0<z<1/2$. Let random variable $X:=\sum_{i=1}^{n/2}x_i$ be the number of uncolored nodes in $V_1$ in round $t^*$. Then, we have $\mathbb{E}[X]\ge (n/2)\cdot (1/\sqrt{n})=\sqrt{n}/2$. Since $x_i$'s are independent, applying the Chernoff bound (see Appendix~\ref{appendix:ineq}) gives us $\Pr[X\le \sqrt{n}/4]\le \exp(-\sqrt{n}/16)=o(1)$.
Hence, w.h.p. after $t^*$ rounds, there is at least one uncolored node in $V_1$, which implies that the process has not converged.

\textbf{Bound in Number of Edges.} We also provide an upper bound in terms of the number of edges in Theorem~\ref{m-bound-thm}, whose proof and tightness are discussed in Appendix~\ref{appendix:m-bound-thm}. The main idea of the proof is to introduce a generalization of traversed node chains.
\begin{theorem}
\label{m-bound-thm}
The convergence time of the \textsc{random Pick} process on a directed graph $G=(V,E)$ is at most $m/\beta$ w.p. $1-\beta$, for any $0<\beta<1$.
\end{theorem}

% \vspace{-0.2cm}
\subsection{Random Initial State}
\label{sec:random-conv}

In a $q$-\textit{random} state, for some $0<q<1$, each node is colored independently w.p. $q$ and uncolored otherwise. We bound the convergence time of the \textsc{Random Pick} process for a $q$-random state in Theorem~\ref{random-convergence-thm}.

% We do not specify red/blue color in the above definition since it does concern the convergence time.
For a graph $G=(V,E)$ and an integer $s$, let $\Gamma_{+}^{(s)}(v):=\{w\in V: d(v,w)\le s\}$ be the $s$-out-neighborhood of node $v\in V$. In particular, $\Gamma_{+}^{(1)}(v)=\Gamma_+(v)\cup \{v\}$. We define $h_q(v)$, for some $0<q<1$, to be the minimum $s$ such that $|\Gamma_{+}^{(s)}(v)|\ge (2\ln n)/q$ and $\infty$ if such $s$ does not exist. We define $\mathcal{H}(q):=\{v\in V: h_q(v)\ne \infty\}$. Now, we present Lemmas~\ref{hit-lemma} and~\ref{high-degree-lemma}, whose proofs follow from some straightforward probabilistic argument. (Complete proofs are given in Appendix~\ref{appendix:hit-lemma} and~\ref{appendix:high-degree-lemma} for the sake of completeness.)

\begin{lemma}
\label{hit-lemma}
In a $q$-random state, every node $v\in \mathcal{H}(q)$ has at least one colored node in its $h_q(v)$-out-neighborhood w.h.p.
\end{lemma}

\begin{lemma}
\label{high-degree-lemma}
In a $q$-random state on a graph $G=(V,E)$, w.p. at least $1-1/n^2$ for every node $v$ in $\{v\in V:d(v)\ge 4(\ln n)^2/q^2\}$, there are at least $d_+(v)q/2$ colored nodes in $\Gamma_{+}(v)$.
\end{lemma}

\begin{theorem}
\label{random-convergence-thm}
Consider the \textsc{Random Pick} process on a graph $G=(V,E)$ with an initial $q$-random state. Then, the convergence time is in $\mathcal{O}((\ln n)^3/q^2)$ w.h.p.
\end{theorem}

\textit{Proof Sketch.}
Let $v$ be an arbitrary node which has a path to a colored node.
We prove that $v$ will be colored in $8(\ln n)^3/q^2$ rounds w.p. at least $1-\mathcal{O}(1/n^2)$. Then, a union bound finishes the proof.

Let $w$ be a colored node whose distance from $v$ is minimized. Consider the node chain $\textbf{w}=w_0, w_1,\cdots, w_h$, where $w_0=w$ and $w_h=v$. If $\textbf{w}$ is traversed, then $v$ is colored. Let $T_w$ be the number of rounds required for $\textbf{w}$ to be traversed. According to Lemma~\ref{trav-chain-lemma}, we have $\mathbb{E}[T_w]\le \sum_{i=1}^{h}d_+(w_i)$. If we prove that $\sum_{i=1}^{h}d_+(w_i)\le 8(\ln n)^2/q^2$, then we can use the same proof as in Theorem~\ref{diam-time-thm} for the directed case to show that w.p. at least $1-1/n^2$, the node chain $\textbf{w}$ is traversed in $32(\ln n)^2(\log_2 n)/q^2=\mathcal{O}((\ln n)^3/q^2)$ rounds. We distinguish between three cases.

\begin{itemize}[leftmargin=10pt,topsep=3pt]
    \item \textbf{Case 1:} $v\notin \mathcal{H}(q)$. The number of nodes reachable from $v$ is less than $(2\ln n)/q$, by the definition of $\mathcal{H}(q)$. This implies that $h$ and $d_+(w_i)$ (for any $1\le i\le h$) are at most $(2\ln n)/q$. Thus, $\sum_{i=1}^{h}d_+(w_i)\le 4(\ln n)^2/q^2$.
    \item \textbf{Case 2:} $v\in \mathcal{H}(q)$ and $h< h_q(v)$. Since $h< h_q(v)$, all $w_i$'s, for $1\le i\le h$, and all their out-neighbors are in distance at most $h_q(v)-1$ from $v$. By the definition of $h_q(v)$, the number of all these nodes is less than $(2\ln n)/q$. This again implies that $h$ and $d_+(w_i)$ (for any $1\le i\le h$) are at most $(2\ln n)/q$. Thus, $\sum_{i=1}^{h}d_+(w_i)\le 4(\ln n)^2/q^2$.
    \item \textbf{Case 3:} $v\in \mathcal{H}(q)$ and $h= h_q(v)$. As in Case 2, we can show that $\sum_{i=2}^{h}d_{+}(w_i)\le 4(\ln n)^2/q^2$. If $d_+(w_1)\le 4(\ln n)^2/q^2$, then we get $\sum_{i=1}^{h}d_{+}(w_i)\le 8(\ln n)^2/q^2$. If $d_+(w_1)> 4(\ln n)^2/q^2$, it suffices to prove that w.p. at least $1-1/n^2$, $w_1$ is colored in $\mathcal{O}((\ln n)^3/q^2)$ rounds. According to Lemma~\ref{high-degree-lemma}, there are at least $d_+(w_1)q/2$ colored nodes in $\Gamma_+(w_1)$. Thus, it is colored in each round w.p. at least $q/2$ independently. The probability that it is not colored after $4(\ln n)^3/q^2$ rounds is at most $(1-q/2)^{4(\ln n)^3/q^2}\le \exp(-2(\ln n)^3/q)\le 1/n^2$.
\end{itemize}
Note we did not consider the case of $v\in \mathcal{H}(q)$ and $h> h_q(v)$ because of Lemma~\ref{hit-lemma}. A full proof is given in Appendix~\ref{appendix:random-convergence-thm}. $\square$

\textbf{Tightness.} Consider a path $v_1,v_2,\cdots,v_n$, where there is an edge from $v_i$ to $v_{i+1}$ for every $1\le i\le n-1$. Furthermore, from each node $v_i$, for $2\le i\le n$, there is an edge to every node which appears before $v_i$ on the path. Consider a $q$-random initial state $\mathcal{S}_0$ for $q=1/\sqrt{n}$. We prove that w.h.p. the convergence time is in $\Omega(1/(q\ln n)^2)$, which demonstrates that the bound in Theorem~\ref{random-convergence-thm} is tight up to some poly-logarithmic term.
% We focus on the case of $q=o(1)$ (because if $q$ is a constant, the bound of $\Omega(1/q^2)=\Omega(1)$ is trivial) and $q=\omega(1/n)$ (because if $q=C/n$ for some constant $C>0$, then w.p. at least $(1-C/n)^{n}\ge (1/4)C$, which is a non-zero constant, there will be no colored node in $\mathcal{S}_0$ and the convergence time is zero.).

\textbf{Claim 1.} \textit{There is no colored node among $\{v_1,\cdots, v_{\tau}\}$ in $\mathcal{S}_0$ for $\tau=\lfloor 1/(q\ln n)\rfloor$ w.h.p.}

\textbf{Claim 2.} \textit{There is at least one colored node in $\mathcal{S}_0$ w.h.p.}

Claims 1 and 2 can be proven using Markov's inequality and the Chernoff bound (see Appendix~\ref{appendix:ineq}), respectively.

Putting Claims 1 and 2 in parallel with the structure of the graph, we can conclude that w.h.p. there is a node $v_h$ such that $v_1$ is colored if and only if the node chain $v_h, v_{h-1},\cdots, v_1$ is traversed and $h\ge \tau\ge \lfloor 1/(q\ln n)\rfloor$. Let $T_i$, for $1\le i\le h-1$, be the number of rounds node $v_i$ needs to pick node $v_{i+1}$ (while nodes $v_{i+1},\cdots, v_h$ have been traversed). Then, the number of rounds for the node chain to be traversed is $T=\sum_{i=1}^{h-1}T_i$. The random variable $T_i$ is geometrically distributed with parameter $1/d_+(v_i)=1/i$, which yields $Var[T_i]=(1-(1/i))/(1/i)^2\le i^2$. Since $T_i$'s are independent $Var[T]=\sum_{i=1}^{h-1}Var[T_i]\le \sum_{i=1}^{h-1}i^2\le h^3$. Applying Chebyshev's inequality (see Appendix~\ref{appendix:ineq}) and using $\mathbb{E}[T]=\sum_{i=1}^{h-1}\mathbb{E}[T_i]=\sum_{i=1}^{h-1}i\ge h^2/4$, we get $\Pr[T\le h^2/8]\le \frac{h^3}{(h^2/8)^2}=64/h$. Since $h\ge \lfloor 1/(q\ln n)\rfloor$, we have $h^2/8\ge \lfloor 1/(q\ln n)\rfloor^2/8$ and $64/h\le 64/\lfloor 1/(q\ln n)\rfloor=\mathcal{O}(\ln n/\sqrt{n})=o(1)$ since $q=1/\sqrt{n}$. Thus, w.h.p. $T\ge \lfloor 1/(q\ln n)\rfloor^2/8=\Omega(1/(q\ln n)^2)$. A more detailed proof is given in Appendix~\ref{appendix:random-convergence-thm}.

\textbf{Remark.} The bound in Theorem~\ref{random-convergence-thm} can be improved to $\mathcal{O}((\ln n)^2/q)$ if we restrict ourselves to undirected graphs. The idea is to prove the bound $\sum_{i=1}^{h}d_+(w_i)\le C\ln n/q$, for some constant $C>0$, instead of $8(\ln n)^2/q^2$ using the property used for the undirected case in the proof of Theorem~\ref{diam-time-thm}.
% A full proof is not provided for the sake of avoiding redundancy.

% \vspace{-0.3cm}
\subsection{Complexity Result}

\begin{definition}[\textsc{Convergence Time} problem]
For a given graph $G=(V,E)$ and two integers $k,t$, determine whether there exists an initial state with exactly $k$ colored nodes such that the expected convergence time is $t$.
\end{definition}

% \vspace{-0.3cm}
\begin{theorem}
\label{thm:hardness-time}
The \textsc{Convergence Time} problem is NP-hard, even when $G$ is undirected and $t=1$.
\end{theorem}
% \vspace{-0.2cm}

Theorem~\ref{thm:hardness-time} builds on a reduction from the \textsc{Vertex Cover} problem, which is known to be NP-hard, cf.~\cite{karp2010reducibility}. A full proof is given in Appendix~\ref{appendix:thm:hardness-time}.

% \vspace{-0.3cm}
\subsection{Convergence Time on Real-world Networks}

So far, we proved some upper bounds which hold for any graph (and any initial state) and are shown to be the best possible, up to some small multiplicative factor. However, stronger bounds for special classes of graphs such as the ones which emerge in the real world might exist. Our experimental findings support stronger bounds for the studied real-world networks and synthetic graphs. While the bounds for these special graphs are tighter than our general bounds from above, they are still ``aligned'' with them.

For each node $v$ in the graph, we ran the process for 300 times starting from the state where only node $v$ is colored and calculated the average convergence time. Then, we reported the minimum, maximum and average value among all nodes in Table~\ref{tab:conv_1}. (Please refer to Section~\ref{sec:exp-alg} for details on our experimental setup.)

\begin{table}
	\centering
		\caption{The maximum, minimum, and average convergence time on different networks. (The suffix of BA stands for the number of nodes.)}\label{tab:conv_1}
  \vspace{-0.4cm}
		%\resizebox{\columnwidth}{!}{
\fontsize{7}{7}\selectfont			
\begin{tabular}{m{0.3cm}<{\centering}m{1.3cm}m{0.85cm}<{\centering}m{0.85cm}<{\centering}m{0.85cm}<{\centering}m{1.7cm}<{\centering}}
				\toprule
                Type &Networks  & Maximum & Minimum & Average & $4D\Delta_+ \log_2 n$ \\
				\midrule
                \multirow{6}*{\rotatebox{90}{Undirected}}
                &Food        & 188.60 & 28.80 & 44.03 &39182\\
                &WikiVote    & 139.64 & 14.36 & 26.10 &35971\\
                &EmailUniv    & 83.21  & 11.15 & 18.57 &23051\\
                &Hamster     & 281.27 & 18.28 & 29.76& 122788\\
                &TVshow     & 166.82 & 41.62 & 56.15& 54097\\
                &Government  & 705.95 & 16.28 & 27.26&285153\\
                \midrule
                \multirow{7}*{\rotatebox{90}{Directed}}
                &Residence & 29.68 & 9.03 & 11.68 & 6333\\
                &FilmTrust & 86.25 & 15.74 & 34.80 & 29978\\
                &BitcoinAlpha &  538.22 &17.46&61.71 & 422166\\
                &BitcoinOTC &  769.33& 19.58&77.80 & 585120\\
                &Gnutella08 &  68.12  &9.41 & 19.77 & 21809\\
                &Advogato   &  807.46 &12.42&34.04 & 430300\\
                &WikiElec   &  849.17& 10.45& 17.91 & 418163\\
                \midrule
                \multirow{4}*{\rotatebox{90}{Synthetic}}
                & BA300 &17.60&7.89&12.79 & 5332\\
                & BA500 &19.74&7.86&13.98 &5379\\
                & BA1000 &25.54&9.37&15.42 & 14510\\
                & BA2000 &33.96&9.45&16.74 & 24563\\        
                % & RGG300 &22.14&11.84&16.24\\
                % & RGG500 &90.58&45.64&66.63\\
                % & RGG1000 &78.77&42.04&57.26\\
                % & RGG2000 &58.6& 32.88 &45.32\\
				\bottomrule
			\end{tabular}
		%}
\end{table}

We observe that the reported convergence times  are always smaller than the upper bound of $4\Delta_+D\log_2 n$ proven in Theorem~\ref{diam-time-thm}. While the maximum convergence times do not match $4\Delta_+D\log_2 n$, they are still closely related to the maximum out-degree of the graph (similar to the general bound) as demonstrated in Figure~\ref{fig:conv}. To quantify this observation more accurately, we have computed the Pearson correlation coefficient between $\Delta_+$ and the maximum convergence time. This is equal to $0.9907$ on the undirected networks, $0.9446$ on the directed networks, and $0.9716$ on BA networks.
Thus, the convergence time is strongly correlated with the maximum out-degree of the graph.

%%%%%%%%%%%%%%%%%%%%%%%%%	
\begin{figure}[t]
		\centering
		\includegraphics[width=0.86\linewidth]{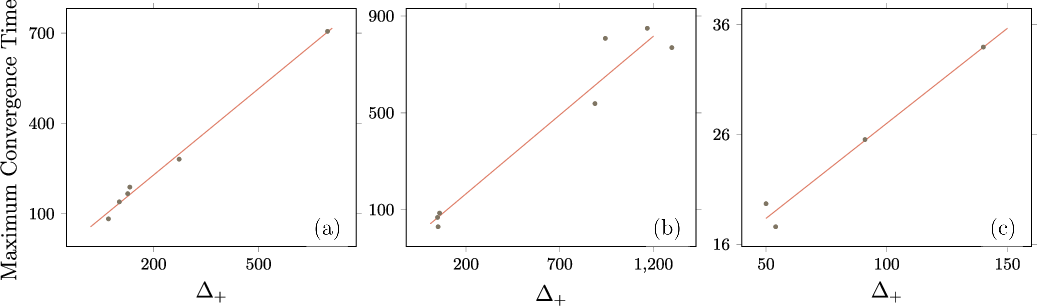}
  \vspace{-0.3cm}
		\caption{The relation between maximum convergence time and $\Delta_+$ on (a) directed, (b) undirected, and (c) BA networks.\label{fig:conv}}
  % \vspace{-0.3cm}
\end{figure}
%%%%%%%%%%%%%%%%%%%%%%%%%%%

We also have conducted experiments for the $q$-random state setup for different values of $q$. For each $q$, we generated 300 $q$-random states and computed the convergence time of the \textsc{Random Pick} process for each state. Then, we reported the average value in Figure~\ref{fig:qstate}. According to the outcome of our experiments, the average convergence time exhibits an inverse linear relationship with $q$ on different networks. More precisely, the Pearson correlation coefficients between the convergence time and $q$ for three networks in Figure~\ref{fig:qstate} are (a) -0.9925, (b) -0.9787, and (c) -0.9783. (See Appendix~\ref{appendix:random-conv} for similar results on other networks.)
The results show a strong inverse linear relationship between the value of $q$ and the convergence time. This is aligned with our theoretical bounds proven in Section~\ref{sec:random-conv}, especially the bound $\tilde{\mathcal{O}}(1/q)$ for undirected graphs.
% However, we do not aim to do a rigorous comparison of these bounds since the bounds in Section~\ref{sec:random-conv} are meant to hold for all graphs and are more relevant when $1/q$ is growing in $n$.

%%%%%%%%%%%%%%%%%%%%%%%%%	
\begin{figure}[t]
		\centering
		\includegraphics[width=0.86\linewidth]{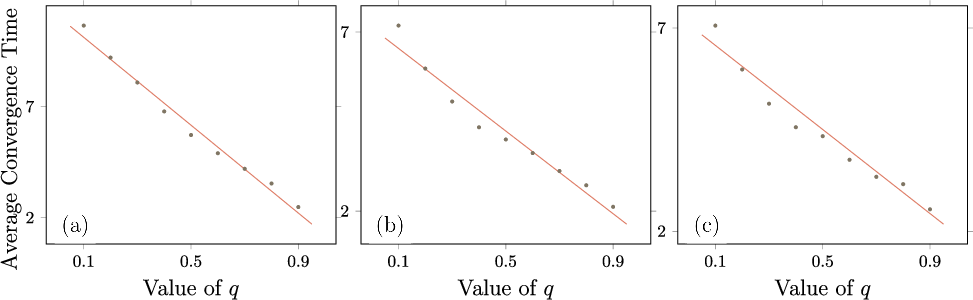}
    \vspace{-0.3cm}
		\caption{The relation between the convergence time and $q$ for (a) Government, (b) WikiElec, and (c) BA2000 networks.\label{fig:qstate}}
  % \vspace{-0.8cm}
\end{figure}
\section{Conclusion}
\label{sec:conclusion}
% \vspace{-0.1cm}

We conducted a comprehensive investigation of the \textsc{Random Pick} process, devised to emulate the adoption of multiple products over a social network. We provided a polynomial time approximation algorithm with the best possible approximation guarantee for the \textsc{Product Adoption Maximization} problem. Furthermore, several tight bounds on the convergence time of the process were proven, in terms of various graph parameters. Additionally, we substantiated our theoretical findings with a plethora of experiments.

As discussed in Section~\ref{sec:greedy}, for fixed $k$ and $\epsilon$, our proposed algorithm runs in $\tilde{\mathcal{O}}(n^3)$ on real-world networks (since $m=\Theta(n)$ and $D$ and $\Delta_+$ are usually small). Despite being almost as fast as baseline algorithms such as highest betweenness/closeness (cf.~\cite{Br01}), it is impractical for deployment on very large networks. Therefore, it would be intriguing to conceive of heuristic algorithms derived from our approach which can cover massive social networks. (For a similar approach applied to the IC model, see~\cite{chen2010scalable}.) One possible course of action would be to consider only a small subset of nodes with the highest degree when determining the $\argmax$ value and set $R$ to a fixed value (as we used $R=300$ in our experiments), thereby mitigating the time complexity to nearly linear time. However, it may be necessary to employ additional approximation techniques to further reduce the hidden constant embedded in $\tilde{\mathcal{O}}$ notation.
 
We have explored the intricacies of the product adoption problem from the standpoint of the red company. A potential avenue for future research is to expand upon the present work by examining the problem through the lens of game theory, where both companies (players) can adapt their respective strategies or choices in response to each other's actions (see~\cite{goyal2012competitive} for some related work).

% Is computing $\mathcal{F}_S(\cdot)$ hard?

% Graphs like complete graph are robust against adversarial attacks. (Ideas from synchronous majority or the adversarial paper. Expansion, vertex-transitivity.)

% \begin{figure}
% \centerline{\includegraphics[height=3.5in]{ecaif01}}
% \caption{Network of transputers and the structure of individual
% processes } \label{procstructfig}
% \end{figure}
\bibliographystyle{ACM-Reference-Format} 
\bibliography{ref}
\newpage~\newpage
\appendix
\section{Some Standard Inequalities}
\label{appendix:ineq}

\begin{theorem}[Chernoff bound]
\label{chenoff}
Suppose that $x_1,\cdots,x_n$ are independently distributed in $[0,1]$ and let $X$ denote their sum, then
\begin{itemize}
    \item $\Pr[(1+\eta)\mathbb{E}[X]\le X]\le \exp\left(-\frac{\eta^2\mathbb{E}[X]}{3}\right)$,
    \item $\Pr[X\le (1-\eta)\mathbb{E}[X]]\le \exp\left(-\frac{\eta^2\mathbb{E}[X]}{2}\right)$.
\end{itemize} 
\end{theorem}

\begin{theorem}[Markov's inequality]
\label{markov}
Let $X$ be a non-negative random variable with finite expectation and $a>0$, then
\[
\Pr[X\ge a]\le \frac{\mathbb{E}[X]}{a}.
\]
\end{theorem}

\begin{theorem}[Chebyshev’s inequality]
\label{chebyshev}
Let $X$ be a random variable with finite variance and $a>0$, then
\[
\Pr[|X-\mathbb{E}[X]|\ge a]\le \frac{Var[X]}{a^2}.
\]
\end{theorem}

\section{Strategy Matters: An Example}
\label{appendix:example}

Consider an undirected star graph with an internal node $v$ and $n-1$ leaves. Let $\lfloor n/2\rfloor-1$ of leaves be blue (and the rest of nodes be uncolored) and the red marketer has the budget to make only 1 node red. If the marketer chooses the internal node $v$, in the \textsc{Random Pick} process, all the remaining $\lceil n/2\rceil$ leaves eventually become red. Thus, there will be more red nodes than blue nodes at the end. However, if the marketer selects a leaf node, almost surely the internal node will eventually become blue (it has significantly more blue out-neighbors than red) and then all the remaining leaves become blue too. Thus, in the second scenario, the process most likely ends with only one red node. This simple example demonstrates the importance of devising effective strategies for accomplishing successful marketing in the \textsc{Product Adoption Maximization} problem.

\section{Proof of Theorem~\ref{hardness}}
\label{appendix:hardness}

The reduction is from the \textsc{Maximum Coverage} problem.

\begin{definition}[\textsc{Maximum Coverage}]
\label{max-cover-def}
Given a collection of subsets $S=\{S_1, S_2,\cdots, S_l\}$ of an element set $O=\{O_1,\cdots, O_h\}$ and an integer $k$, what is the maximum number of elements covered by $k$ subsets? An element is covered if it is in at least one of the selected subsets.
\end{definition}
Without loss of generality, we can assume that each element $O_i$ appears in at least one subset (otherwise, it could be simply ignored.)

\begin{theorem}[cf.~\cite{feige1998threshold}]
\label{mc-hardness}
There is no polynomial time $(1-\frac{1}{e})$-approximation algorithm for the \textsc{Maximum Coverage} problem, unless $NP \subseteq DTIME(n^{\mathcal{O}(\log \log n)})$.
\end{theorem}

\textbf{Transformer.} Now, let us define our transformer which transforms an instance of the \textsc{Maximum Coverage} problem to an instance of our problem. Let us construct the graph $G=(V,E)$ in a step-by-step manner. Consider the node sets $V_S:=\{s_1,\cdots, s_l\}$, $V_O:=\{o_1,\cdots, o_h\}$. Then, add edge $(o_i, s_j)$, for $1\le i\le h$ and $1\le j\le l$, if and only if $O_i\in S_j$. To complete the construction, for each node $o_i$, add $\lceil 1/\epsilon \rceil -1$ distinct nodes and add an edge from all these nodes to $o_i$. See Figure~\ref{hardness-fig}, for an example. Furthermore, consider the state $\mathcal{S}$ where all nodes are uncolored and let the budget be equal to $k$.

\begin{figure}

\centerline{\includegraphics[height=2in]{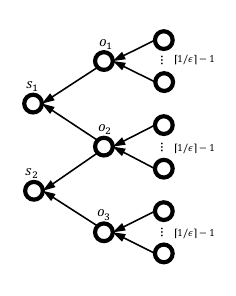}}
% \vspace{-0.5cm}
\caption{The outcome of the transformer for $S=\{S_1, S_2\}$, where $S_1=\{O_1,O_2\}$ and $S_2=\{O_2,O_3\}$, and $O=\{O_1,O_2,O_3\}$. \label{hardness-fig}}
% \vspace{-0.5cm}
\end{figure}

\begin{observation}
\label{obs-Vs}
Let $A$ be a seed set of size $k$ in an instance of the \textsc{Production Adoption Maximization} problem constructed by the above transformer. Then, there is a seed node $A'$ of size $k$ (or smaller) which produces a solution of the same size (or larger) and its intersection with nodes in $V_O$ and their attached leaves is empty (i.e., all nodes of $A'$ are in $V_S$). 
\end{observation}
It is straightforward to see the correctness of Observation~\ref{obs-Vs}. If a node $o_i$ or one of its attached leaves is in $A$, we can choose a node $w\in \Gamma_+(o_i)$ (which is in $V_S$) and add $w$ to the seed set instead. (Note that $\Gamma_+(o_i)\ne \emptyset$ because element $o_i$ appears in at least one subset.) We observe that if $w$ is colored red, eventually $o_i$ and all its leaves are colored red. Note that since there is no blue node, the order in which the nodes are colored red does not change the final outcome.

\textbf{Connection Between Optimal Solutions.} For an arbitrary instance of the \textsc{Maximum Coverage} problem, let $OPT_{MC}$ denote the optimal solution. Let $OPT_{PAM}$ be the optimal solution of our problem for the instance built by the above transformer. Then, we have
\begin{equation}
\label{opt-eq}
OPT_{MC}=\frac{OPT_{PAM}-k}{\lceil 1/\epsilon\rceil}.
\end{equation}
Consider a collection of subsets $S'\subset S$ of size $k$ which covers $OPT_{MC}$ elements. Let $V_{S'}$ be the corresponding subset of $V_{S}$. If we select $V_{S'}$ as the seed set, for each element $O_i$ which is covered by $S'$, node $o_i$ and all its $\lceil 1/\epsilon\rceil-1$ leaves will become red. This gives us a solution of size at least $k+OPT_{MC}\lceil 1/\epsilon\rceil$. Thus, $OPT_{PAM}\ge k+OPT_{MC}\lceil 1/\epsilon\rceil$.

Consider a seed set $A$ of size $k$ which gives a solution of size $OPT_{PAM}$. We can assume that all nodes of $A$ are in $V_S$, according to Observation~\ref{obs-Vs}. Based on the construction, a node $o_i$ and its attached leaves will eventually become red if and only if $o_i$ has an edge to a node in $A$. This implies that $(OPT_{PAM}-k)/\lceil 1/\epsilon\rceil$ nodes in $V_O$ have at least one out-neighbor in $A$. Thus, the $k$ subsets in $S$ corresponding to the nodes in $A$ cover at least $(OPT_{PAM}-k)/\lceil 1/\epsilon\rceil$ elements, which yields $OPT_{MC}\ge (OPT_{PAM}-k)/\lceil 1/\epsilon\rceil$.

\textbf{Inapproximability.} Assume there is a polynomial time $(1-1/e+\epsilon)$-approximation algorithm $\mathcal{A}lg$, for some constant $\epsilon>0$, which solves the \textsc{Production Adoption Maximization} problem. For an instance of the \textsc{Maximum Coverage} problem, we can use the aforementioned transformer to construct and instance of our problem in polynomial time. Then, we use the algorithm $\mathcal{A}lg$ to solve the constructed instance. Let $Sol_{PAM}$ be the outcome. Then, using an argument similar to the second part of the proof of Equation~(\ref{opt-eq}), we have a solution $Sol_{MC}$ for the \textsc{Maximum Coverage} problem such that $Sol_{MC}\ge (Sol_{PAM}-k)/\lceil 1/\epsilon \rceil$. Using $Sol_{PAM}\ge (1-1/e+\epsilon)OPT_{PAM}$ and then Equation~(\ref{opt-eq}), we get
\begin{align*}
Sol_{MC}\ge & \frac{Sol_{PAM}-k}{\lceil 1/\epsilon\rceil} \ge \frac{(1-1/e+\epsilon) OPT_{PAM}-k}{\lceil 1/\epsilon \rceil} \\ = & (1-1/e+\epsilon)OPT_{MC}+\frac{(\epsilon-1/e)k}{\lceil 1/\epsilon \rceil}.
\end{align*}
Using the fact that $OPT_{MC}\ge k$ and some small calculations, we get $Sol_{MC}\ge (1-1/e)OPT_{MC}$. Thus, we get a $(1-1/e)$-approximation algorithm for the \textsc{Maximum Coverage} problem, which we know is not possible according to Theorem~\ref{mc-hardness}, unless $NP \subseteq DTIME(n^{\mathcal{O}(\log \log n)})$.

\section{An Example of Extended Sequences}
\label{appendix:es-ex}

Consider the graph and initial state given in Figure~\ref{model-example}. Let the pick sequences be as follows
\begin{align*}
ps(v_1)= v_3, v_2, v_2, \cdots, ps(v_2)= v_1, v_4, v_1, \cdots, \\ ps(v_3)=v_4,v_4,v_4,\cdots, ps(v_4)= v_5,v_5, v_5, \cdots.
\end{align*}
We observe that $es^0(v_1)=v_1$, $es^1(v_1)=es^0(v_1)$, $es^0(v_3)=v_1, v_3$, and $es^2(v_1)=es^1(v_1)$, $es^1(v_2)= v_1, v_3, v_2, v_1$, and $es^3(v_1)=es^2(v_1)$, $es^2(v_2)= v_1, v_3, v_2, v_1, v_2, v_1, v_4, v_5$. Therefore, $es(v_2)=v_1, v_3, v_2,$ $v_1, v_2, v_1, v_4, v_5,\cdots$. According to Lemma~\ref{extended-seq-lemma}, $v_2$ will be colored red since $v_4$ is the first node which appears in the extended sequence of $v_2$ and is colored in the initial state. Also, if we follow the process with the given pick sequences: (1) $v_1$ picks $v_3$ and $v_2$ picks $v_1$ and both remain uncolored, (2) in the next round $v_1$ remains uncolored since it picks $v_2$, but $v_2$ becomes red since it picks $v_4$, (3) one round after, $v_1$ picks $v_2$ and thus becomes red too. 

\section{Proof of Lemma~\ref{extended-seq-lemma}}

\label{appendix:lemma-es}

The proof is based on induction on $t$. For the base case of $t=0$, $es^0(v)=v$ for which the statement is trivial.
% $es^0(v),es^1(v)=v, v'$, where $v'$ is the first node in $ps(v)$ (the random choice of $v$ in the first round). It is clear that if both $v$ and $v'$ are uncolored in $\mathcal{S}_0$, then $\mathcal{S}_1(v)=u$. If $v$ is red/blue, it clearly remains blue/red; otherwise, it picks the color of $v'$ in the first round. Thus, $\mathcal{S}_1(v)$ is equal to the color of the first uncolored node in the sequence $es^0(v), es^{1}(v)$ and $\mathcal{S}_1(v)=u$ if such a node does not exist.

Now, assume that the statement is true for some $t\ge 0$, we prove that it also holds for $t+1$. Recall that $es^{t+1}(v)=es^{t}(v),es^{t}(v')$, where $v'$ is the $(t+1)$-th element in $ps(v)$. We consider the following three cases:
\begin{itemize}
    \item \textbf{Case 1:} All nodes in $es^{t+1}(v)$ are uncolored in $\mathcal{S}_0$. Thus, there is no colored node in $es^{t}(v)$ and $es^{t}(v')$. By the induction hypothesis, both $v$ and $v'$ are uncolored at round $t$. Since $v$ picks $v'$ in round $t+1$, it will remain uncolored, that is, $\mathcal{S}_{t+1}(v)=u$.
    \item \textbf{Case 2:} There is at least one node in $es^{t}(v)$ which is colored in $\mathcal{S}_0$. Let $w$ be the first colored node in $es^{t}(v)$. Then, by the induction hypothesis $\mathcal{S}_{t}(v)=\mathcal{S}_0(w)$. Since a colored node does not change its color, we have $\mathcal{S}_{t+1}(v)=\mathcal{S}_0(w)$.
    \item \textbf{Case 3:} There is at least one node in $es^{t}(v')$ which is colored initially, but there is no such node in $es^{t}(v)$. Let $w$ be the first colored node in $es^{t}(v')$. By the induction hypothesis, $\mathcal{S}_t(v)=u$ and $\mathcal{S}_t(v')=\mathcal{S}_0(w)$. Since $v$ picks $v'$ in round $t+1$, we have $\mathcal{S}_{t+1}(v)=\mathcal{S}_t(v')=\mathcal{S}_0(w)$.
\end{itemize}

\section{Table of Experimented Networks}
\label{appendix:table}

\begin{table}
	\centering
		\caption{Some properties of the experimented networks, where the suffix of BA represents the number of nodes.}\label{tab:graphs}
		%\resizebox{\columnwidth}{!}{
\fontsize{8}{8}\selectfont			
\begin{tabular}{m{0.4cm}<{\centering}m{1.8cm}m{1.1cm}<{\centering}m{1.2cm}<{\centering}m{0.5cm}<{\centering}m{0.9cm}<{\centering}}
				\toprule
                Type &Networks & Nodes & Edges & $D$ & $\Delta_+$  \\
				\midrule
                \multirow{6}*{\rotatebox{90}{Undirected}}
                &Food      & 620  &  2091 & 8  & 132  \\
                &WikiVote  & 889  &  2914 & 9  & 102  \\
                &EmailUniv & 1133 &  5451 & 8  & 71   \\
                &Hamster   & 2426 & 16630 & 10 & 273  \\
                &TVshow    & 3892 & 17239 & 9  & 126  \\
                &Government& 7057 & 89429 & 8  & 697 \\
                \midrule
                \multirow{7}*{\rotatebox{90}{Directed}}
                &Residence   & 217  & 2672   &  4& 51 \\
                &FilmTrust   & 874  & 1853   &  13& 59  \\
                &BitcoinAlpha& 3783 & 24186  & 10&888\\
                &BitcoinOTC &  5881 & 35592  & 9 &1298\\
                &Gnutella08 &  6301 & 20777  & 9 & 48 \\
                &Advogato   &  6541 & 51127  & 9 & 943\\
                &WikiElec   &  7118 & 103675 & 7 &1167\\
                \midrule
                \multirow{4}*{\rotatebox{90}{Synthetic}}
                & BA300 & 300 &891 & 3&54\\
                & BA500 & 500 &1491 & 3&50\\
                & BA1000 &1000&2991 & 4&91\\
                & BA2000 &2000&5991 & 4&140 \\
                % & RGG300 &300 &338 & 8 & 4\\
                % & RGG500 &500 &957& 37&11\\
                % & RGG1000 &1000&3751&22&17\\
                % & RGG2000 &2000&15330&21&28\\
				\bottomrule
			\end{tabular}
		%}
\end{table}

\section{Additional Diagrams: Comparison of Algorithms}
\label{appendix:alg-diagram}
 We provide diagrams similar to the ones in Figure~\ref{fig:2} for some other networks. We also give Table~\ref{tab:time} which reports the exact value of average final ratio of red nodes for each algorithm, including the standard deviations.
 
%%%%%%%%%%%%%%%%%%%%%%%%%	
\begin{figure}[h]
		\centering
		\includegraphics[width=1\linewidth]{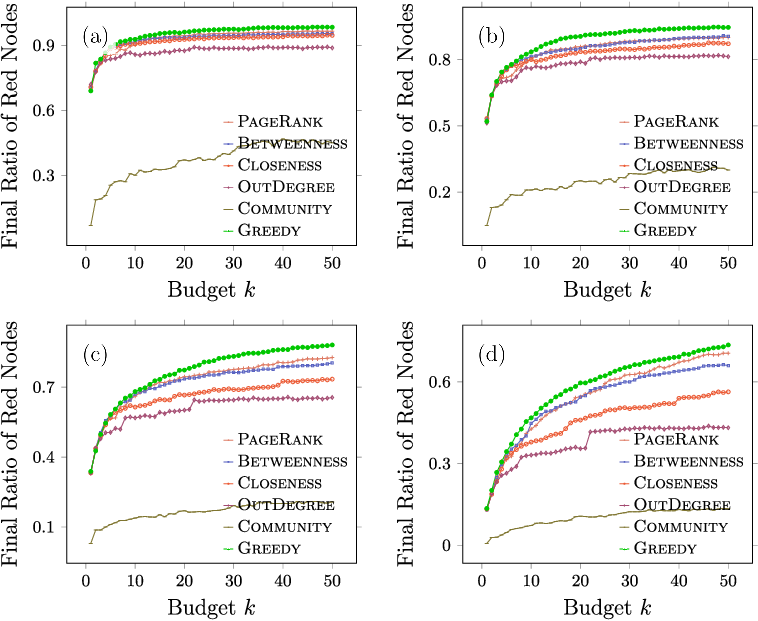}
    % \vspace{-0.4cm}
		\caption{The final ratio of red nodes for $k$ seed nodes and $b_0=$ (a) 10, (b) 20, (c) 50 and (d) 100 initial blue nodes on the TVshow network. \label{fig:4}}
  % \vspace{-0.4cm}
\end{figure}
%%%%%%%%%%%%%%%%%%%%%%%%%%%

%%%%%%%%%%%%%%%%%%%%%%%%%	
\begin{figure}[h]
		\centering
		\includegraphics[width=1\linewidth]{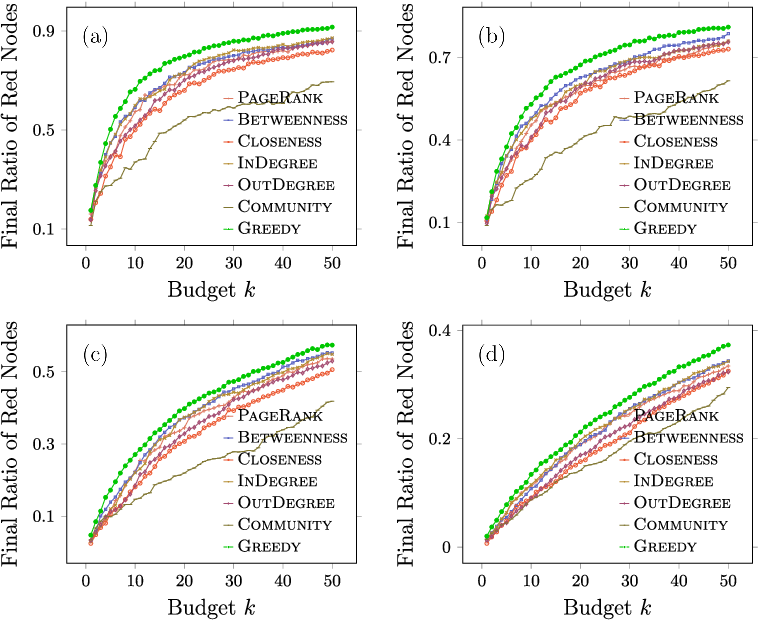}
      % \vspace{-0.4cm}
		\caption{The final ratio of red nodes for $k$ seed nodes and $b_0=$ (a) 10, (b) 20, (c) 50 and (d) 100 initial blue nodes on the Residence network. \label{fig:1}}
  % \vspace{-0.4cm}
\end{figure}
%%%%%%%%%%%%%%%%%%%%%%%%%%%

%%%%%%%%%%%%%%%%%%%%%%%%%	
\begin{figure}[h]
		\centering
		\includegraphics[width=1\linewidth]{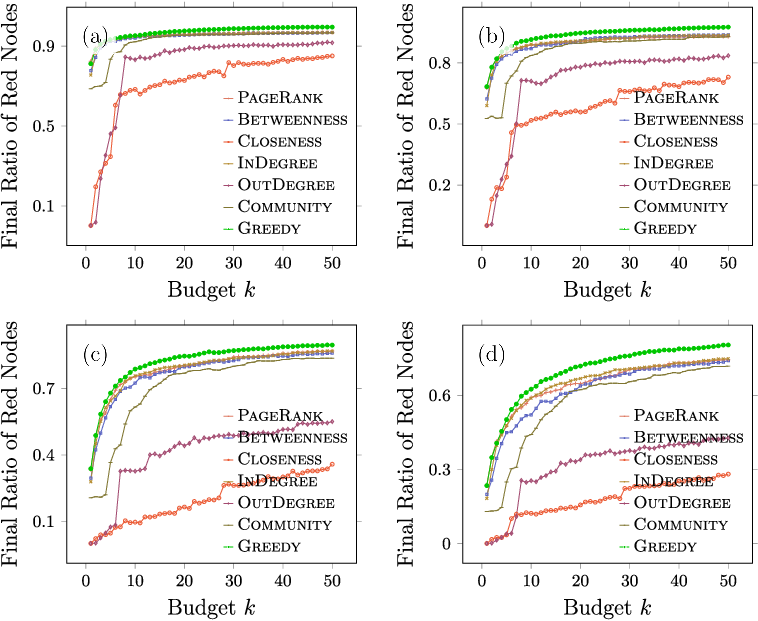}

        % \vspace{-0.4cm}
		\caption{The final ratio of red nodes for $k$ seed nodes and $b_0=$ (a) 10, (b) 20, (c) 50 and (d) 100 initial blue nodes on the Advogato network. \label{fig:3}}
  % \vspace{-0.4cm}
\end{figure}
%%%%%%%%%%%%%%%%%%%%%%%%%%%

%%%%%%%%%%%%%%%%%%%%%%%%	
\begin{figure}[h]
		\centering
		\includegraphics[width=1\linewidth]{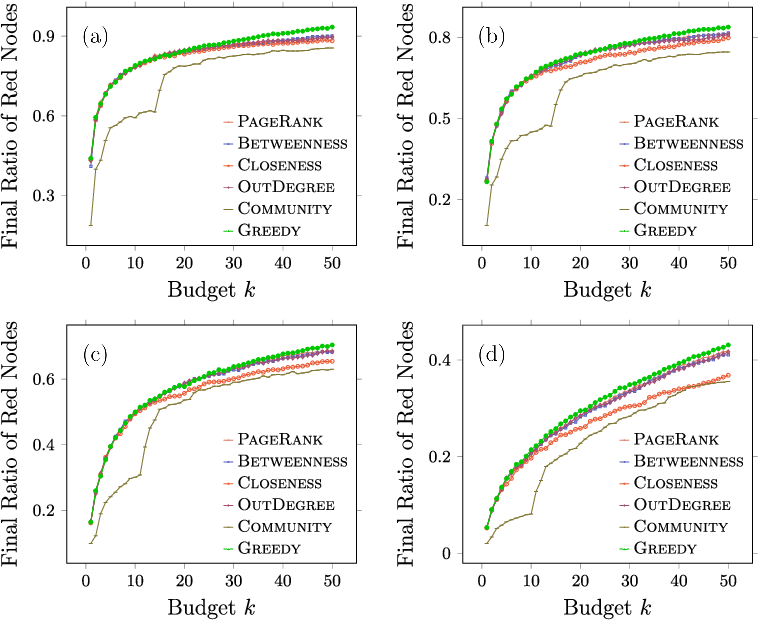}

          % \vspace{-0.4cm}
		\caption{The final ratio of red nodes for $k$ seed nodes and $b_0=$ (a) 10, (b) 20, (c) 50 and (d) 100 initial blue nodes on the BA300 network. \label{fig:5}}
  % \vspace{-0.4cm}
  
\end{figure}
%%%%%%%%%%%%%%%%%%%%%%%%%%

%%%%%%%%%%%%%%%%%%%%%%%%	
\begin{figure}[h]
		\centering
		\includegraphics[width=1\linewidth]{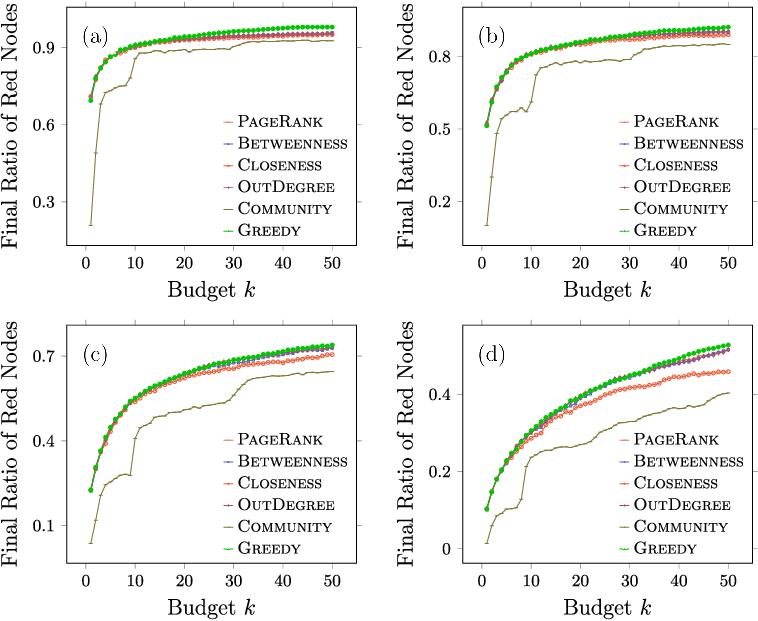}
  % \vspace{-0.4cm}
		\caption{The final ratio of red nodes for $k$ seed nodes and $b_0=$ (a) 10, (b) 20, (c) 50 and (d) 100 initial blue nodes on the BA500 network. \label{fig:6}}

\end{figure}
%%%%%%%%%%%%%%%%%%%%%%%%%%

\begin{table*}
	\centering
		\caption{The final ratio of red nodes for different algorithms for $k=50$. The values in parentheses correspond to the standard deviations.}\label{tab:time}
  %\resizebox{\columnwidth}{!}{
\fontsize{8}{8}\selectfont			
\begin{tabular}{m{1cm}<{\centering}m{0.6cm}<{\centering}m{1.5cm}<{\centering}m{1.5cm}<{\centering}m{1.8cm}<{\centering}m{1.5cm}<{\centering}m{1.5cm}<{\centering}m{1.5cm}<{\centering}m{1.5cm}<{\centering}}
				\toprule
                \multirow{2}*{Networks} &\multirow{2}*{$b_0$} & \multicolumn{7}{c}{Final ratio of red nodes} \\
                \cmidrule{3-9} && \textsc{Greedy} &  \textsc{PageRank}& \textsc{Betweenness}&\textsc{Closeness}& \textsc{InDegree}& \textsc{OutDegree}& \textsc{Community} \\
                \midrule

                \multirow{4}*{Food} 
                & 10 &\textbf{0.983(0.007)}  &0.907(0.018) & 0.941(0.009) & 0.758(0.058) & -- & 0.901(0.024) & 0.745(0.050)  \\
                & 20 &\textbf{0.925(0.014)}  &0.862(0.019) & 0.874(0.025) & 0.731(0.063) & -- & 0.856(0.021) & 0.706(0.059) \\
                & 50 &\textbf{0.810(0.022)}  &0.768(0.023) & 0.752(0.030) & 0.566(0.037) & -- & 0.704(0.027) & 0.545(0.032)   \\
                &100 & \textbf{0.605(0.022)}  &0.585(0.027) & 0.506(0.043) & 0.409(0.025) & -- & 0.530(0.027) & 0.341(0.030) \\
                     \midrule
                 \multirow{4}*{TVshow} 
                & 10 &\textbf{0.983(0.011)}  &0.965(0.012) & 0.954(0.012) & 0.946(0.019) & -- & 0.888(0.052) & 0.450(0.158)  \\
                & 20 &\textbf{0.945(0.017)}  &0.900(0.016) & 0.905(0.020) & 0.871(0.026) & -- & 0.813(0.041) & 0.299(0.099) \\
                & 50 &\textbf{0.880(0.020)}  &0.826(0.021) & 0.804(0.027) & 0.734(0.047) & -- & 0.655(0.035) & 0.207(0.084)   \\
                &100 & \textbf{0.735(0.026)}  &0.705(0.024) & 0.659(0.028) & 0.563(0.036) & -- & 0.432(0.048) & 0.134(0.029) \\
           
                \midrule
               \multirow{4}*{Residence} 
                & 10 &\textbf{0.915(0.019)}  &0.858(0.033) & 0.868(0.021) & 0.822(0.043) & 0.872(0.030) & 0.856(0.034) & 0.694(0.067)  \\
                & 20 &\textbf{0.810(0.029)}  &0.759(0.031) & 0.786(0.028) & 0.735(0.041) & 0.759(0.035) & 0.755(0.037) & 0.615(0.061)  \\
                & 50 &\textbf{0.573(0.031)}  &0.535(0.035) & 0.552(0.029) & 0.505(0.034) & 0.548(0.030) & 0.529(0.038) & 0.418(0.033)  \\
                &100 & \textbf{0.373(0.019)}  &0.334(0.019) & 0.343(0.019) & 0.325(0.018) & 0.342(0.017) & 0.323(0.018) & 0.294(0.017)  \\
                \midrule
                \multirow{4}*{Advogato} 
                & 10 &\textbf{0.995(0.003)}  &0.967(0.006) & 0.968(0.006) & 0.850(0.007) & 0.968(0.006) & 0.916(0.028) & 0.965(0.007)  \\
                & 20 &\textbf{0.975(0.003)}  &0.930(0.013) & 0.937(0.015) & 0.731(0.086) & 0.935(0.011) & 0.835(0.044) & 0.927(0.013) \\
                & 50 &\textbf{0.895(0.012)}  &0.869(0.014) & 0.858(0.018) & 0.358(0.087) & 0.868(0.014) & 0.550(0.084) & 0.835(0.025)   \\
                &100 & \textbf{0.804(0.016)}  &0.744(0.020) & 0.740(0.024) & 0.281(0.068) & 0.748(0.019) & 0.427(0.071) & 0.718(0.030) \\
                \midrule
                \multirow{4}*{BA300} 
                & 10 &\textbf{0.934(0.012)}  &0.898(0.017) & 0.901(0.017) & 0.882(0.021)& -- & 0.895(0.019)  & 0.855(0.031)  \\
                & 20 &\textbf{0.838(0.021)}  &0.813(0.022) & 0.816(0.021) & 0.792(0.022) & --& 0.810(0.022)  & 0.745(0.035) \\
                & 50 &\textbf{0.704(0.022)}  &0.687(0.019) & 0.682(0.023) & 0.654(0.024) & --& 0.686(0.021)  & 0.629(0.029)   \\
                &100 & \textbf{0.431(0.019)}  &0.418(0.017) & 0.410(0.019) & 0.368(0.020) & --& 0.414(0.019)  & 0.355(0.020) \\
                 \midrule
                \multirow{4}*{BA500} 
                & 10 &\textbf{0.979(0.001)}  &0.954(0.010) & 0.958(0.008) & 0.950(0.013) & --& 0.955(0.009)  & 0.927(0.027)  \\
                & 20 &\textbf{0.921(0.014)}  &0.902(0.014) & 0.902(0.014) & 0.890(0.019) & --& 0.903(0.014)  & 0.849(0.034) \\
                & 50 &\textbf{0.739(0.024)}  &0.729(0.022) & 0.728(0.020) & 0.705(0.025) & -- & 0.728(0.020) & 0.645(0.037)   \\
                &100 & \textbf{0.528(0.022)}  &0.517(0.023) & 0.515(0.021) & 0.458(0.024)& -- & 0.513(0.022)  & 0.403(0.031) \\
                % \midrule
                %  \multirow{4}*{RGG500} 
                % & 10 &\textbf{0.961(0.001)}  &0.855(0.024) & 0.771(0.030) & 0.645(0.041) & -- & 0.792(0.034) & 0.796(0.028)  \\
                % & 20 &\textbf{0.840(0.019)}  &0.713(0.027) & 0.579(0.047) & 0.352(0.032) & -- & 0.590(0.039) & 0.566(0.031) \\
                % & 50 &\textbf{0.626(0.025)}  &0.549(0.027) & 0.407(0.022) & 0.296(0.027) & -- & 0.446(0.030) & 0.435(0.026)   \\
                % &100 & \textbf{0.392(0.020)}  &0.339(0.017) & 0.261(0.024) & 0.215(0.010) & -- & 0.314(0.018) & 0.305(0.019) \\
				\bottomrule
			\end{tabular}
		%}
\end{table*}

\section{Proof and Tightness of Theorem~\ref{m-bound-thm}}
\label{appendix:m-bound-thm}

In this section, we provide the proof of Theorem~\ref{m-bound-thm} and discuss its tightness, building on a generalization of a traversed node chain.

\textbf{Generalized Traversed Node Chain.} We first need to generalize the concept of a traversed node chain. Let us replace each node $w_i$ in the node chain $w_0,w_1,\cdots,w_h$ with $w_{i,1},w_{i,2},\cdots, w_{i,l_i}$. Furthermore, we define a map $\mathcal{M}$ which maps a node $w_{i,j}$, for $1\le i\le h$ and $1\le j\le l_i$, to a node $\mathcal{M}(w_{i,j})$ which appears before $w_{i,j}$ in the sequence. We require $(w_{i,j},\mathcal{M}(w_{i,j}))$ to be in $E$. In the \textsc{Random Pick} process, we say that the generalized node chain $w_{0,1},w_{0,2},\cdots,w_{0,l_0},\cdots, w_{h,1},w_{h,2},\cdots, w_{h,l_h}$ is \textit{traversed} if there are time steps $t_{i,j}$'s such that $t_{i,j}<t_{i',j'}$ if $w_{i,j}$ appears before $w_{i',j'}$ in the node chain and $ps_{t_{i,j}}(w_{i,j})=\mathcal{M}(w_{i,j})$. (Note that if $l_i=1$, for $0\le i\le h$, and $\mathcal{M}(w_{i,1})=w_{i-1,1}$, for $1\le i\le h$, then we retrieve the original node chain definition for the node chain $w_{0,1},w_{1,1},\cdots, w_{h,1}$.) It is straightforward to observe that the statement of Lemma~\ref{trav-chain-lemma} holds for the generalized node chain with the bound $\sum_{i=1}^{h}\sum_{j=1}^{l_i}d_+(w_{i,j})$. 

Furthermore, recall that in a graph $G=(V,E)$, $d(v,w)$ for $v,w\in V$ is the length of the shortest path from $v$ to $w$ (and it is $\infty$ if there is no such path). For a node set $A\subset V$, we define $d(v,A):=\min_{w\in A} d(v,w)$.

\textbf{Proof of Theorem~\ref{m-bound-thm}.}
Consider an arbitrary initial state $\mathcal{S}_0$. Let $V'$ be the set of nodes which have a path to a colored (red/blue) node. We want to bound the number of rounds it takes until all nodes in $V'$ are colored. Let $L_0=R_0\cup B_0$ be the set of initially colored nodes. We set $h:=\max_{v\in V'\setminus L_0} d(v,L_0)$. Partition the nodes in $V'\setminus L_0$ into subsets $L_i$'s, for $1\le i\le h$, where $L_i:=\{v\in V'\setminus L_0: d(v,L_0)=i\}$ is the nodes whose distance to $L_0$ is $i$. Label the nodes in $L_i$ (for $0\le i\le h$) from $w_{i,1}$ to $w_{i,l_i}$ in an arbitrary order, where $l_i:=|L_i|$.

Let us consider the generalized node chain $\textbf{w}=w_{0,1},w_{0,2},$ $\cdots, w_{0,l_0},$ $\cdots, w_{h,1}, w_{h,2},\cdots, w_{h,l_h}$. Let $\mathcal{M}$ map each node $w_{i,j}\in L_i$ to a node $w_{i-1,j'}\in L_{i-1}$. Note that such a node $w_{i-1,j'}$ must exist by definition and $w_{i-1,j'}$ appears before $w_{i,j}$ in the chain. Thus, using the extension of Lemma~\ref{trav-chain-lemma}, which we stated before, the expected number of rounds for the node chain $\textbf{w}$ to be traversed is equal to $\sum_{i=1}^{h}\sum_{j=1}^{l_i}d_+(w_{i,j})$. Since each edge in $E$ is counted in the sum at most once, then $\sum_{i=1}^{h}\sum_{j=1}^{l_i}d_+(w_{i,j})\le m$.

Using a simple inductive argument, we can prove that if the node chain $\textbf{w}$ is traversed, then all nodes on the node chain (which is equal to $V'$) are colored. Therefore, the convergence time is not larger than the number of rounds for $\textbf{w}$ to be traversed. Let $T$ denote the convergence time of the process, then using $\mathbb{E}[T]\le m$ and Markov's inequality (see Appendix~\ref{appendix:ineq}) we get
\[
\Pr[T\ge m/\beta]\le \Pr[T\ge \mathbb{E}[T]/\beta]\le \frac{\beta\mathbb{E}[T]}{\mathbb{E}[T]}=\beta
\]
for any $0<\beta<1$. $\square$ 

\textbf{Tightness.} Based on Theorem~\ref{m-bound-thm}, the convergence time is in $\mathcal{O}(m)$ (with an arbitrarily small constant error probability). This could be quadratic in $n$ for dense graphs. To support its tightness, we provide an example graph and initial state for which the convergence time is in $\Omega(n^2)$ w.h.p.

Consider a path $w_0,w_1,\cdots, w_{n/2}$, where there is an edge from $w_i$ to $w_{i-1}$ for any $1\le i\le n/2$. Furthermore, consider a node set $W'$ of size $n/2-1$ and assume that there is an edge from every node $w_i$ to every node in $W'$. Consider the \textsc{Random Pick} process with the initial state $\mathcal{S}_0$, where all nodes are uncolored except $w_0$. Nodes in $W'$ remain uncolored forever. However, $w_1$ to $w_{n/2}$ are colored one after another. More precisely, once $w_1$ picks $w_0$, it becomes colored. Then, all nodes remain unchanged until $w_2$ picks $w_1$ and becomes colored. And this continues until $w_{n/2}$ is colored. We observe that the number of rounds the process needs to converge is equal to the number of rounds required to traverse the node chain $w_0,w_1,\cdots, w_{n/2}$.

Based on Lemma~\ref{trav-chain-lemma}, we have $\mathbb{E}[T]=\sum_{i=1}^{n/2}d_{+}(v)=(n/2)\cdot (n/2)=n^2/4$, where $T$ is the convergence time. Let $T_i$ for $1\le i\le n/2$ be the number of rounds $w_i$ needs to pick $w_{i-1}$ (while $w_0,\cdots, w_{i-1}$ has been traversed). We have $T=\sum_{i=1}^{n/2}T_i$. Random variable $T_i$ is geometrically distributed with parameter $1/d_+(w_i)=2/n$. From this, we get $Var[T_i]=(1-(2/n))/(2/n)^2\le n^2/4$. Since $T_i$'s are independent, we have
\[
Var[T]=\sum_{i=1}^{n/2}Var[T_i]\le (n/2)\cdot (n^2/4)=n^3/8.
\]
Applying Chebyshev’s inequality (see Appendix~\ref{appendix:ineq}) and using $\mathbb{E}[T]=n^2/4$ and $Var[T]=n^3/8$, we have
\[
\Pr[T\le n^2/8]\le \Pr[|T-\frac{n^2}{4}|\ge \frac{n^2}{8}]\le \frac{n^3/8}{(n^2/8)^2}=8/n.
\]
Therefore, w.p. at least $1-8/n$, the process takes at least $n^2/8=\Omega(n^2)$ rounds to converge.

\section{Proof of Lemma~\ref{hit-lemma}}
\label{appendix:hit-lemma}
Consider an arbitrary node $v$ in $\mathcal{H}(q)$. By definition, the $h_q(v)$-out-neighborhood of $v$ contains at least $(2\ln n)/q$ nodes. The probability that none of them is colored is at most $(1-q)^{(2\ln n)/q}\le \exp(-2\ln n)=1/n^2$, where we used the estimate $1-z\le \exp(-z)$. Applying a union bound implies that our desired statement holds w.p. at least $1-1/n$.

\section{Proof of Lemma~\ref{high-degree-lemma}}
\label{appendix:high-degree-lemma}
Label the out-neighbors of $v$ from $w_1$ to $w_{d_+(v)}$. Define the Bernoulli random variable $x_i$ to be 1 if and only if node $w_{i}$ is colored. Let $X:=\sum_{i=1}^{d_+(v)}x_i$ be the number of colored nodes in $\Gamma_+(v)$. Then, we have $\mathbb{E}[X]=d_+(v)q\ge 4(\ln n)^2/q$. Since $x_i$'s are independent, we can use the Chernoff bound (see Appendix~\ref{appendix:ineq}) which yields
\begin{align*}
\Pr[X\le d_+(v)q/2]=\Pr[X\le \mathbb{E}[X]/2]\le \exp(-\mathbb{E}[X]/8)\\ \le \exp(-(\ln n)^2/(2q))\le 1/n^2.
\end{align*}
A union bound over all possible choices of node $v$ completes the proof.

\section{Proof and Tightness of Theorem~\ref{random-convergence-thm}}
\label{appendix:random-convergence-thm}

\textbf{Proof of Theorem~\ref{random-convergence-thm}.}
Let $\mathcal{E}_1$ and $\mathcal{E}_2$, respectively, be the event that the statement of Lemma~\ref{hit-lemma} and Lemma~\ref{high-degree-lemma} hold. We know that both these events happen w.h.p. Thus, we assume they hold in the remaining of the proof. To be fully accurate, we need to condition on $\mathcal{E}_1$ (analogously $\mathcal{E}_2$) happening in the probabilities and expectations below, but we skip that for the sake of readability.

Let $v$ be an arbitrary uncolored node which has a path to a colored node.
We prove that $v$ will be colored in $8(\ln n)^3/q^2$ rounds w.p. at least $1-\mathcal{O}(1/n^2)$. Then, a union bound finishes the proof.

Let $w$ be a colored node whose distance from $v$ is minimized. Consider the node chain $\textbf{w}=w_0, w_1,\cdots, w_h$, where $w_0=w$ and $w_h=v$, corresponding to a shortest path from $v$ to $w$ of length $h$. By a simple inductive argument, we can prove that if $\textbf{w}$ is traversed, then $v$ is colored. Let $T_w$ be the number of rounds required for $\textbf{w}$ to be traversed. According to Lemma~\ref{trav-chain-lemma}, we have $\mathbb{E}[T_w]\le \sum_{i=1}^{h}d_+(w_i)$. If we prove that $\sum_{i=1}^{h}d_+(w_i)\le 8(\ln n)^2/q^2$, then we can use the same proof as in Theorem~\ref{diam-time-thm} for the directed case to prove that w.p. at least $1-1/n^2$, the node chain $\textbf{w}$ is traversed in $32(\ln n)^2(\log_2 n)/q^2=\mathcal{O}((\ln n)^3/q^2)$ rounds (by basically replacing $D\Delta_+$ with $8(\ln n)^2/q^2$). We need to distinguish between three cases.

\begin{itemize}
    \item \textbf{Case 1:} $v\notin \mathcal{H}(q)$. The number of nodes reachable from $v$ is less than $(2\ln n)/q$, by the definition of $\mathcal{H}(q)$. This implies that $h$ and $d_+(w_i)$ (for any $1\le i\le h$) are at most $(2\ln n)/q$. Therefore, $\sum_{i=1}^{h}d_+(w_i)\le 4(\ln n)^2/q^2$.
    \item \textbf{Case 2:} $v\in \mathcal{H}(q)$ and $h< h_q(v)$. Since $h< h_q(v)$, all $w_i$'s, for $1\le i\le h$, and all their out-neighbors are in distance at most $h_q(v)-1$ from $v$. By the definition of $h_q(v)$, the summation of all these nodes is less than $(2\ln n)/q$. This again implies that $h$ and $d_+(w_i)$ (for any $1\le i\le h$) are at most $(2\ln n)/q$. Therefore, $\sum_{i=1}^{h}d_+(w_i)\le 4(\ln n)^2/q^2$.
    \item \textbf{Case 3:} $v\in \mathcal{H}(q)$ and $h= h_q(v)$. Note that following the same argument as in Case 2, we can show that $\sum_{i=2}^{h}d_{+}(w_i)\le 4(\ln n)^2/q^2$. If $d_+(w_1)\le 4(\ln n)^2/q^2$, then we get $\sum_{i=1}^{h}d_{+}(w_i)\le 8(\ln n)^2/q^2$. In that case, we are done. If $d_+(w_1)> 4(\ln n)^2/q^2$, it suffices to prove that w.p. at least $1-1/n^2$, $w_1$ is colored in $\mathcal{O}((\ln n)^3/q^2)$ rounds (because then we can again use $\sum_{i=2}^{h}d_+(w_i)\le 4(\ln^2 n)/q^2$ to show that $w_1,w_2,\cdots, w_h$ will be traversed in $\mathcal{O}((\ln n)^3/q^2)$ rounds w.p. at least $1-1/n^2$, which gives overall probability of $1-2/n^2$). According to Lemma~\ref{high-degree-lemma} (we are conditioning on event $\mathcal{E}_2$ as mentioned above), there are at least $d_+(w_1)q/2$ colored nodes in $\Gamma_+(w_1)$. Thus, it is colored in each round w.p. at least $q/2$ independently. The probability that it is not colored after $4(\ln n)^3/q^2$ rounds is at most $(1-q/2)^{4(\ln n)^3/q^2}\le \exp(-2(\ln n)^3/q^2)\le 1/n^2$ (using $1-z\le \exp(-z)$).
\end{itemize}

Note that we did not consider the case of $v\in \mathcal{H}(q)$ and $h> h_q(v)$ because of Lemma~\ref{hit-lemma} (i.e., event $\mathcal{E}_1$).

\textbf{Tightness of Theorem~\ref{random-convergence-thm}.}
Consider a path $v_1,v_2,\cdots,v_n$, where there is an edge from $v_i$ to $v_{i+1}$ for every $1\le i\le n-1$. Furthermore, there is an edge from each node $v_i$, for $2\le i\le n$, to every node which appears before $v_i$ on the path (i.e., $v_j$ for $j<i$). Consider the \textsc{Random Pick} process on this graph with a $q$-random initial state $\mathcal{S}_0$ for $q=1/\sqrt{n}$. We prove that w.h.p. the convergence time is in $\Omega(1/(q\ln n)^2)$ which asserts that the bound in Theorem~\ref{random-convergence-thm} is tight, up to some polylogarithmic term. Let us first prove the following two claims.
% We focus on the case of $q=o(1)$ (because if $q$ is a constant, the bound of $\Omega(1/q^2)=\Omega(1)$ is trivial) and $q=\omega(1/n)$ (because if $q=C/n$ for some constant $C>0$, then w.p. at least $(1-C/n)^{n}\ge (1/4)C$, which is a non-zero constant, there will be no colored node in $\mathcal{S}_0$ and the convergence time is zero.).

\textbf{Claim 1.} \textit{There is no colored node among $\{v_1,\cdots, v_{\tau}\}$ in $\mathcal{S}_0$ for $\tau=\lfloor 1/(q\ln n)\rfloor$ w.h.p.}

Let $X$ be the number of colored nodes in $\{v_1,\cdots, v_{\tau}\}$. Using Markov's inequality (see Appendix~\ref{appendix:ineq}) and $\mathbb{E}[X]=\tau q\le 1/\ln n$, we have $\Pr[X\ge 1]\le 1/\ln n=o(1)$. Thus, $X=0$ w.h.p.

\textbf{Claim 2.} \textit{There is at least one colored node in $\mathcal{S}_0$ w.h.p.}

Let $Y$ be the number of colored nodes. Then, using $\mathbb{E}[Y]=nq=\omega(1)$ and Chernoff bound (Appendix~\ref{appendix:ineq}), we have $\Pr[Y=0]=o(1)$.

Putting Claims 1 and 2 in parallel with the structure of the graph, we can conclude that w.h.p. there is a node $v_h$ such that $v_1$ is colored if and only if the node chain $v_h, v_{h-1},\cdots, v_1$ is traversed and $h\ge \tau\ge \lfloor 1/(q\ln n)\rfloor$. Let $T_i$, for $1\le i\le h-1$ be the number of rounds node $v_i$ needs to pick node $v_{i+1}$ (while nodes $v_{i+1},\cdots, v_h$ have been traversed). Then, the number of rounds for the node chain to be traversed is $T=\sum_{i=1}^{h-1}T_i$. The random variable $T_i$ is geometrically distributed with parameter $1/d_+(v_i)=1/i$, which yields $Var[T_i]=(1-(1/i))/(1/i)^2\le i^2$. Since $T_i$'s are independent $Var[T]=\sum_{i=1}^{h-1}Var[T_i]\le \sum_{i=1}^{h-1}i^2\le h^3$. Applying Chebyshev's inequality (see Appendix~\ref{appendix:ineq}) and using $\mathbb{E}[T]=\sum_{i=1}^{h-1}\mathbb{E}[T_i]=\sum_{i=1}^{h-1}i\ge h^2/4$, we get 
\[
\Pr[T\le h^2/8]\le \frac{h^3}{(h^2/8)^2}=64/h.
\]
Since $h\ge \lfloor 1/(q\ln n)\rfloor$, we have $h^2/8\ge \lfloor 1/(q\ln n)\rfloor^2/8$ and $64/h\le 64/\lfloor 1/(q\ln n)\rfloor=\mathcal{O}(\ln n/\sqrt{n})=o(1)$ since $q=1/\sqrt{n}$. Thus, w.h.p. we have $$T\ge h^2/8 \ge \lfloor 1/(q\ln n)\rfloor^2/8=\Omega(1/(q\ln n)^2).$$

\section{Proof of Theorem~\ref{thm:hardness-time}}
\label{appendix:thm:hardness-time}

The reduction is from the \textsc{Vertex Cover} problem, which is known to be NP-hard, cf.~\cite{karp2010reducibility}.

For an undirected connected graph $G=(V,E)$, a node set $V'\subset V$ is a \emph{vertex cover} if for every edge $(v,u)\in E$, either $v$ or $u$ is in $V'$. For a given undirected graph $G$ and integer $k<n$, the goal of the \textsc{Vertex Cover} problem is to determine whether there is a vertex cover of size $k$ in $G$ or not.

We claim that for an undirected connected graph $G$, there is an initial state with $k$ colored nodes whose expected convergence time is $1$ if and only if $G$ has a vertex cover of size $k$. This implies that the \textsc{Convergence Time} problem is NP-hard.

Let $V'\subset V$ be a vertex cover of size $k$ in $G$. Consider the initial state where all nodes in $V'$ are colored (and nodes in $V\setminus V'$ are uncolored). Every node $w\in V\setminus V'$ has at least one neighbor (since $G$ is connected) and all its neighbors are colored (since $V'$ is a vertex cover). Thus, all nodes in $V\setminus V'$ (which is non-empty since $k<n$) will be colored deterministically in the next round, which implies that the convergence time is $1$.

Consider an initial state where only nodes in $V'$ of size $k$ are colored and the expected convergence time is equal to $1$. We claim that $V'$ is a vertex cover. For the sake of contradiction, assume that there is an edge $(v,w)$ such that $v,w\in V\setminus V'$. Then, there is a non-zero probability that $v$ picks $w$ in the first round, which results in a convergence time larger than 1. Note that the convergence time cannot be 0. (This is true since the graph is connected and $k<n$. Thus, there is at least one node which will be colored and that takes at least 1 round.) Therefore, the expected convergence time is strictly larger than 1, which is a contradiction.This implies that $V'$ is indeed a vertex cover.

\section{Additional Diagrams for Convergence Time}
\label{appendix:random-conv}

Here, we provide diagrams similar to the ones in Figure~\ref{fig:qstate} for other networks. We also have computed the Pearson  correlation coefficients between the average convergence time and the value of $q$ for all 14 networks in Figure~\ref{fig:qstate_appen}, which are (a) -0.9370, (b) -0.9845, (c) -0.9831, (d) -0.9823, (e) -0.9461, (f) -0.9790, (g) -0.9670, (h) -0.9654, (i) -0.9596, (j) -0.9900, (k) -0.9874, (l) -0.9834, (m) -0.9873, and (n) -0.9901.
%%%%%%%%%%%%%%%%%%%%%%%%%	
\begin{figure*}[t]
		\centering
		\includegraphics[width=1\linewidth]{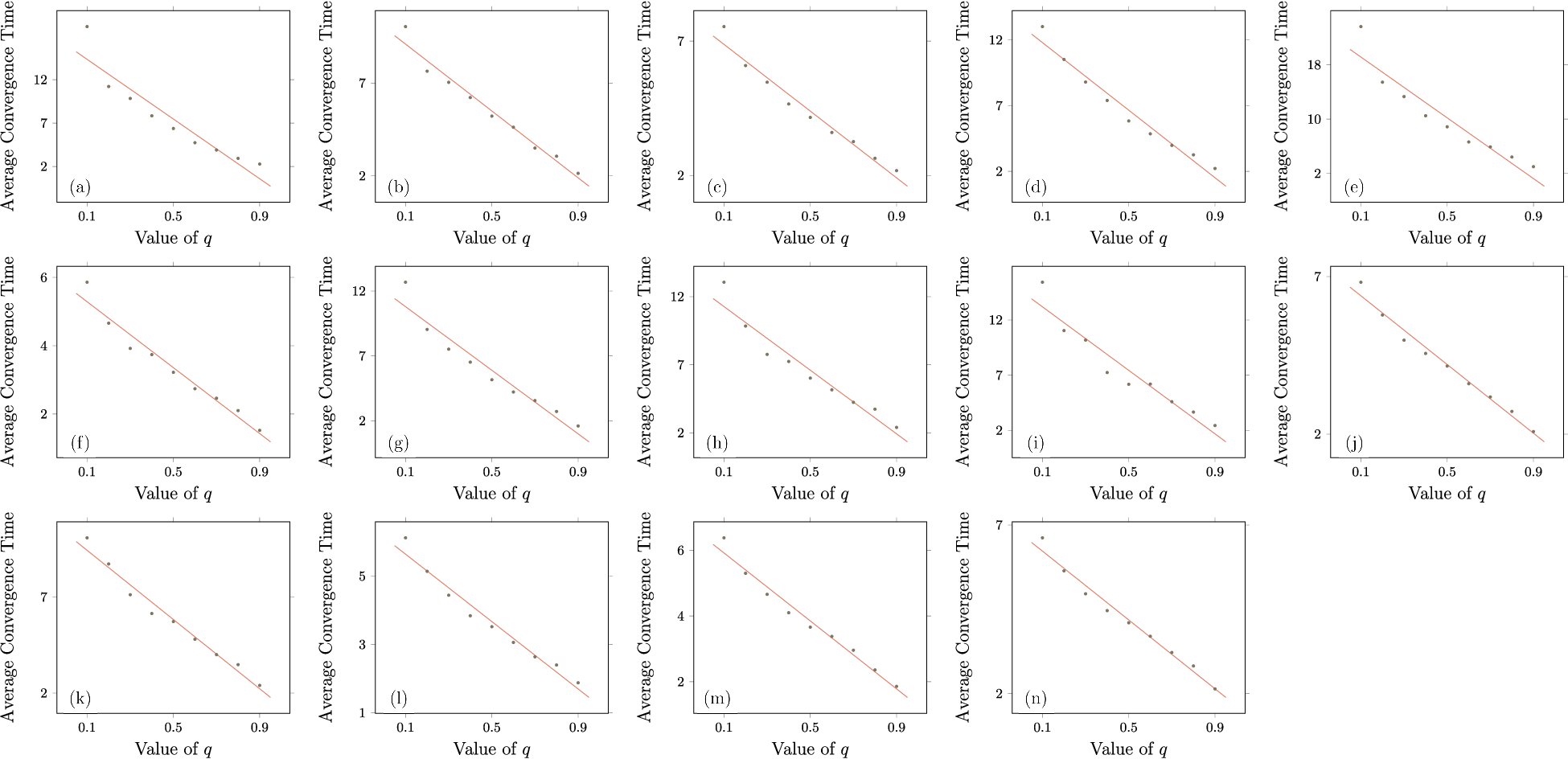}
		\caption{The relation between the convergence time and $q$ for (a) Food, (b) WikiVote, (c) EmailUniv, (d) Hamster, (e) TVshow, (f) Residence, (g) FilmTrust, (h) BitcoinAlpha, (i) BitcoinOTC, (j) Gnutella08, (k) Advogato, (l) BA300, (m) BA500, and (n) BA1000 networks.\label{fig:qstate_appen}}
\end{figure*}
%(a) -0.9831, (b) -0.9925, (c) -0.9655, (d) -0.9787, (e) -0.9900, (f) -0.9783, (g) -0.9274, and (h) -0.9636. 
%%%%%%%%%%%%%%%%%%%%%%%%%%%

%\ahad{The diagrams similar to Figure~\ref{fig:qstate} for other networks will go here.}
\end{document}